\documentclass{amsart}
\pdfoutput=1
 \usepackage{amssymb}
\usepackage{graphicx}
\usepackage{enumerate}

\usepackage{marginnote}

\usepackage{amstext}
 \usepackage{floatrow}%
 \usepackage{paralist}
\usepackage{mathrsfs}  

   \newtheorem{lem}{\protect\lemmaname}
   \newtheorem{theorem}{Theorem}
   \newtheorem{remark}{Remark}
  \newtheorem*{septhm}{Theorem \ref{th:sep-proof-idea}}

\newcommand{\boundedGraph}{\ensuremath{G}}

\newcommand{\blank}{\ensuremath{\mathrm{blank}}}

\newcommand{\colors}{\ensuremath{\mathrm{colors}}}

\newcommand{\functionFormula}{\ensuremath{\delta}}
\newcommand{\containmentFormula}{\ensuremath{\epsilon}}
\newcommand{\definitionFormula}{\ensuremath{\zeta}}
\newcommand{\colorFormula}{\ensuremath{\eta}}
\newcommand{\gaifmannFormula}{\ensuremath{\theta}}

\newcommand{\bounded}{\ensuremath{\mathit{bnd}}}
\newcommand{\unbounded}{\ensuremath{\mathit{unb}}}

\newcommand{\tvlc}{\ensuremath{{C}^2}}
\newcommand{\mso}{\mathrm{MSO}}

\newcommand{\gaifmannRelation}{\ensuremath{R}}
\newcommand{\signatureOrientedKTree}{\ensuremath{\mathcal{R}}}
\newcommand{\signatureDefs}{\ensuremath{\mathcal{A}}}

\newcommand{\signatureCTwo}{{\ensuremath{\mathcal{C}_\unbounded}}}
\newcommand{\signatureMSO}{{\ensuremath{\mathcal{C}_\bounded}}}

\newcommand{\signatureCTwoExtended}{{\ensuremath{\mathcal{D}_\unbounded}}}
\newcommand{\signatureMSOextended}{{\ensuremath{\mathcal{D}_\bounded}}}
\newcommand{\signatureMain}{\ensuremath{\mathcal{E}}}
\newcommand{\signatureP}{{\ensuremath{\mathcal{P}}}}

\newcommand{\copySymbol}[1]{\ensuremath{\overline{#1}}}
\newcommand{\signatureIntersection}{\ensuremath{\mathcal{B}}}

\newcommand{\intersectionRelation}{\ensuremath{B}}
\newcommand{\intersectionRelationCopy}{{\ensuremath{\copySymbol{\intersectionRelation}}}}
\newcommand{\formulaCTwo}{\ensuremath{\beta}}
\newcommand{\formulaMSO}{\ensuremath{\alpha}}
\newcommand{\formulaMSOcopy}{\ensuremath{\psi}}
\newcommand{\formulaCTwoExtended}{\ensuremath{\beta'}}
\newcommand{\formulaMSOextended}{\ensuremath{\alpha'}}
\newcommand{\model}{\ensuremath{\mathfrak{M}}}
\newcommand{\modelA}{\ensuremath{\mathfrak{A}}}
\newcommand{\modelUniverse}{\ensuremath{M}}
\newcommand{\modelAlt}{\ensuremath{\mathfrak{L}}}
\newcommand{\modelFinal}{\ensuremath{\mathfrak{N}}}

\newcommand{\modelHat}{{\ensuremath{\widehat{\model}}}}
\newcommand{\modelTilde}{{\ensuremath{\widetilde{\model}}}}
\newcommand{\rank}{\ensuremath{\mathit{rank}}}

\newcommand{\oneType}{\ensuremath{\pi}}
\newcommand{\twoType}{\ensuremath{\lambda}}
\newcommand{\Types}{\ensuremath{\text{-Types}}}
\newcommand{\MsgTypes}[1]{\ensuremath{#1\text{-MsgTypes}}}
\newcommand{\oneTypes}{\ensuremath{\text{1}\Types}}
\newcommand{\twoTypes}{\ensuremath{\text{2}\Types}}

\newcommand{\messageTypeNumber}{\ensuremath{l}}
\newcommand{\messageTypeCTwo}{\ensuremath{S}}
\newcommand{\messageAlphabetCTwo}{\ensuremath{\mathcal{S}}}
\newcommand{\messageType}{\ensuremath{T}}
\newcommand{\messageAlphabet}{\ensuremath{\mathcal{T}}}

\newcommand{\oneTypeOf}{\ensuremath{\text{1-} \mathfrak{tp}}}
\newcommand{\twoTypeOf}{\ensuremath{\text{2-} \mathfrak{tp}}}
\newcommand{\typeFormula}{\ensuremath{\mathfrak{char}}}

\newcommand{\typeFormulaConjunct}{\ensuremath{\iota}}
\newcommand{\adjunctFormula}{\ensuremath{\mu}}

\newcommand{\degree}{\ensuremath{\mathit{deg}}}
\newcommand{\messageGraph}{\ensuremath{G}}
\newcommand{\outDegree}{\ensuremath{\degree^+}}

\newcommand{\ar}{\mathit{arity}}

\newcommand{\hin}{\mathrm{hin}}
\newcommand{\down}{\mathrm{inv}}
\newcommand{\self}{\mathrm{self}}

\newcommand{\transLONG}{\mathsf{\mathbf{interpret}}}
\newcommand{\trans}{\mathsf{\mathbf{i}}}
\newcommand{\tr}{\mathsf{\mathbf{tr}}}

\newcommand{\dpath}{\mathrm{dpath}}
\newcommand{\gaif}{\mathrm{Gaif}}
\newcommand{\MS}{\mathrm{MSO}}
\newcommand{\CMS}{\mathrm{CMSO}}
\newcommand{\reach}{\mathrm{reach}}
\newcommand{\reduLONG}{\mathsf{\mathbf{undecorate}}}
\newcommand{\redu}{\mathsf{\mathbf{u}}}
\newcommand{\rt}{\mathrm{root}}
\newcommand{\qr}{\mathrm{qr}}
\newcommand{\labLONG}{\mathsf{\mathbf{structurize}}}
\newcommand{\lab}{\mathsf{\mathbf{s}}}
\newcommand{\unary}{\mathrm{un}}
\newcommand{\binary}{\mathrm{bin}}

\renewcommand{\part}{\mathrm{part}}
\newcommand{\leaf}{\mathrm{leaf}}
\newcommand{\internal}{\mathrm{int}}

\renewcommand{\t}{\mathsf{\mathbf{t}}}

\newcommand{\Op}{\stackrel{\to}{\circ}}

\newcommand{\allleaves}{\mathrm{leaves}}
\newcommand{\allinternal}{\mathrm{ints}}
\newcommand{\thisver}{\mathrm{this}}

\newcommand{\intdist}{\mathrm{intdist}}
\newcommand{\intsame}{\mathrm{intsame}}
\newcommand{\child}{\mathrm{child}}
\newcommand{\twochildren}{\mathrm{children}}

\newcommand{\tw}{\mathrm{tw}}

\newcommand{\WSone}{\mathrm{WS1S}}
\newcommand{\inj}{\mathrm{inj}}
\newcommand{\dom}{\mathrm{dom}}
\newcommand{\img}{\mathrm{img}}
\newcommand{\suc}{\mathrm{suc}}
\renewcommand{\u}{\mathsf{\mathbf{u}}}
\newcommand{\Pdef}{\mathrm{suc}\mbox{-}\mathrm{rel}}

\newcommand{\card}{\mathrm{card}}

\newcommand{\psihin}{\Theta}

\newcommand{\enc}{\mathrm{enc}}

\newcommand{\D}{\mathrm{Label}}
\newcommand{\arb}{\mathrm{arb}}

\newcommand{\cardLogic}{\mathrm{\MS^{\exists card}}}

\newcommand{\mfA}{\mathfrak{A}}
\newcommand{\mfB}{\mathfrak{B}}
\newcommand{\mfM}{\mathfrak{M}}
\newcommand{\mfN}{\mathfrak{N}}
\newcommand{\mfO}{\mathfrak{O}}
\newcommand{\mfT}{\mathfrak{T}}
\newcommand{\mfP}{\mathfrak{P}}
\newcommand{\mfR}{\mathfrak{R}}

\renewcommand{\P}{\mfP}
\newcommand{\M}{\mfM}

\newcommand{\mcN}{\mathcal{N}}
\newcommand{\mcR}{\mathcal{R}}

\newcommand{\mcC}{\mathcal{C}}
\newcommand{\mcD}{\mathcal{D}}
\newcommand{\mcU}{\mathcal{U}}

\newcommand{\mcV}{\mathcal{V}} 

\renewcommand{\Xi}{\mcN}

\newcommand{\K}{\mathscr{K}}
\renewcommand{\H}{\mathscr{H\scriptstyle{IN}}}

\usepackage[english]{babel}
\providecommand{\lemmaname}{Lemma}

\title{Monadic second order finite satisfiability and unbounded tree-width
}

\author{Tomer Kotek, Helmut Veith, Florian Zuleger}
\address{TU Vienna}
\email{\{kotek,zuleger@forsyte.at\}}

\makeatletter
      \def\@setcopyright{}
      \def\serieslogo@{}
\makeatother
   
\begin{document}

\begin{abstract}
The finite satisfiability problem of monadic second order logic is
decidable only on classes of structures of bounded tree-width by the
classic result of Seese~\cite{ar:Seese91}.  We prove that the following problem
is decidable: 
\begin{quote}
\textbf{Input:} (i) A monadic second order logic sentence $\alpha$,
and (ii) a sentence $\beta$ in the two-variable fragment of first
order logic extended with counting quantifiers. The vocabularies of
$\alpha$ and $\beta$ may intersect. {\small \par}

\textbf{Output:} Is there a finite structure which satisfies $\alpha\land\beta$
such that the restriction of the structure to the vocabulary of $\alpha$
has bounded tree-width? (The tree-width of the desired structure is
not bounded.){\small \par}
\end{quote}
As a consequence, we prove the decidability of the satisfiability
problem by a finite structure of bounded tree-width of a logic $\cardLogic$ extending
monadic second order logic with linear cardinality constraints of
the form $|X_{1}|+\cdots+|X_{r}|<|Y_{1}|+\cdots+|Y_{s}|$
on the variables $X_i$, $Y_j$ of the outer-most quantifier block. 
We prove the decidability of a similar extension of WS1S. 
\end{abstract}

\maketitle

 \section{Introduction}

Monadic second order logic ($\MS$) is among the most expressive logics with good algorithmic properties. It has found countless
applications in computer science in diverse areas ranging from verification and automata theory~\cite{bk:mona,ar:Madhusudan05,bk:Thomas97}
to combinatorics~\cite{ar:KotekMakowskyLMCS14,bk:Lovasz12}, and parameterized complexity theory~\cite{bk:GroheFlum06,bk:CourcelleEngelFriet12}. 

The power of $\MS$ is most visible over graphs of bounded tree-width, and with second order quantifiers ranging over sets of edges\footnote{The 
logic we denote by $\MS$ is denoted $\mathrm{MS}_2$ by Courcelle and Engelfriet~\cite{bk:CourcelleEngelFriet12}.}: 
(1) Courcelle's famous theorem shows that $\MS$ model checking is decidable over graphs of bounded tree-width in linear time~\cite{ar:Courcelle90,ar:ArnborgEtAl}. 
(2) Finite satisfiability by graphs of bounded tree-width is decidable~\cite{ar:Courcelle90} (with non-elementary complexity) -- thus contrasting Trakhtenbrot's 
undecidability result of first order logic.
(3) Seese proved~\cite{ar:Seese91} that for each class $\K$ of graphs with \emph{unbounded} tree-width, 
finite satisfiability of $\MS$ by graphs in $\K$ is undecidable. 
Together, (2) and (3) give a fairly clear picture of the decidability of finite satisfiability of $\MS$.
It appeared that (3) gives a natural limit for decidability of $\MS$ on graph classes. 
For instance, finite satisfiability on planar graphs is undecidable because their tree-width is unbounded.

While Courcelle and Seese circumvent Trakhtenbrot's undecidability result by \emph{restricting the classes of graphs (or relational structures)}, 
several other research communities have studied syntactic restrictions of first order logic. Modal logic~\cite{ar:Vardi96}, 
many temporal logics~\cite{ar:Pnueli77},~\cite[Chapter 24]{handbook-AR}, the guarded fragment~\cite{ar:Gradel99}, 
many description logics~\cite{handbook-DL}, and the two-variable fragment~\cite{ar:GO99} are restricted first order logics with 
decidable finite satisfiability, and hundreds of papers on these topics have explored the border between decidability and undecidability.
While many of the earlier papers exploited variations of the tree model property to show decidability, recent research has also focused on logics 
such as the two-variable fragment with counting $C^2$~\cite{ar:GOR97,ar:PH05}, where finite satisfiability is decidable 
despite the absence of the tree model property. In a recent breakthrough result, Charatonik and Witkowski~\cite{pr:CW13} have extended this result to structures 
with built-in binary trees. 
Note that this logic is not a fragment of first order logic, but more naturally understood as a very weak second order logic which can express one specific 
second order property -- the property of being a tree.

Our main result is a powerful generalization of the seminal result on decidability of the satisfiability problem of $\MS$ over bounded tree-width and 
the recent theorem by~\cite{pr:CW13}: We show decidability of finite satisfiability of conjunctions 
$\alpha \wedge \beta$ where $\alpha$ is in $\MS$ and $\beta$ is in $C^2$ by a finite structure $\M$ 
whose restriction to the vocabulary of $\alpha$ has bounded tree-width. (Theorem~\ref{th:main} in Section~\ref{se:overview})

Let us put this result into perspective:

\begin{itemize}
\item The $\MS$ decidability problem is a trivial consequence by setting $\beta$ to true; Charatonik and Witkowski's result follows by choosing $\alpha$ 
to be an $\MS$ formula which axiomatizes a $d$-ary tree, which is a standard construction~\cite{bk:CourcelleEngelFriet12}.
\item The decidability of \emph{model checking} $\alpha \wedge \beta$ over a finite structure is a much simpler problem than ours: 
We just have to model check $\alpha$ and $\beta$ one after the other. In contrast, satisfiability is not obvious because 
$\alpha$ and $\beta$ can share relational variables.
running two finite satisfiability algorithms for the two formulas independently
may tield two models which disagree on the shared vocabulary. Thus, the problem we consider is similar in spirit 
to (but technically very different from) Nelson-Oppen~\cite{ar:NO79} combinations of theories.
\item Our result trivially generalizes to Boolean combinations of sentences in the two logics.
\end{itemize}

\subsection*{Proof Technique.}
We show how to reduce our satisfiability problem for  $\alpha\land\beta$ 
to the finite satisfiability of a $C^2$-sentence with a built-in tree, which is decidable by~\cite{pr:CW13}. 
The most significant technical challenge is to eliminate shared binary relation symbols between
$\alpha'$ and $\beta$. Our Separation Theorem overcomes this challenge by an elegant construction 
based on local types of universe elements and a coloring argument for directed graphs. 
The second technical challenge is
to  replace the $\MS$-sentence $\alpha$ with an equi-satisfiable 
$C^2$-sentence $\alpha'$. To do so, we apply tools including the Feferman-Vaught theorem for $\MS$ and translation schemes. 

\subsection*{Monadic Second Order Logic with Cardinalities.}
Our main theorem imply new decidability results for monadic second order logic with cardinality constraints, i.e., expressions of the form
$|X_1| + \ldots |X_r| < |Y_1| + \ldots |Y_t|$ where the $X_i$ and $Y_i$ are monadic second order variables. 
Klaedtke and Rue\ss~\cite{ar:KR03} showed that the decision problem for the theory of weak monadic second order logic with cardinality constraints of 
one successor ($\WSone^{\card}$) is undecidable; 
they describe a decidable fragment where the second order quantifiers have no alternation and appear \emph{after} the first order quantifiers in the prefix. 
Our main theorem implies decidability of a different fragment of WS1S with cardinalities: The fragment $\cardLogic$ consists of formulas $\exists \bar{X} \psi$ where 
the cardinality constraints in $\psi$ involve only the monadic second order variables from $\bar{X}$, cf.~Theorem~\ref{th:cardinality} in Section~\ref{se:cardinality}. 
Note that in contrast to~\cite{ar:KR03}, our fragment is a strict superset of WS1S.

For WS2S, we are not aware of results about decidable fragments with cardinalities. We describe a strict superset 
of $\MS$ whose satisfiability problem over finite graphs of bounded tree-width is decidable, and which is syntactically similar to the WS1S extension above.


\subsection*{Expressive Power over Structures.} 
Our main result extends the existing body of results on finite satisfiability by structures of bounded tree-width 
to a significantly richer set of structures. The structures we consider are $C^2$-axiomatizable 
extensions of structures of bounded tree-width. For instance, 
we can have interconnected doubly-linked lists as in Fig.~\ref{fig:expressive-graphs}(a), or a tree whose leaves are connected in a chain
and have edges pointing to 
any of the nodes of a cyclic list as in Fig.~\ref{fig:expressive-graphs}(b).
Such structures occur very naturally as shapes of dynamic data structures in programming 
-- where cycles and trees are containers for data, and additional edges express relational information between the data. 
The analysis of semantic relations between data structures served as 
a motivation for us to investigate the logics in the current paper~\cite{pr:iFM14}.

Being a cyclic list or a tree whose leaves are chained can be expressed in $\MS$
and both of these data structures have tree-width at most $3$. 
We can compel the edges between the tree and the cyclic list to obey 
$C^2$-expressible constraints such as: 
\begin{itemize}
\item every leaf of the tree has a single edge to the cyclic list;
\item every node of the cyclic list has an incoming edge from at least one 
leaf of the tree; or
\item any two leaves pointing to the same node of the cyclic list agree on membership in some unary relation. 
\end{itemize}
Note that while the structures we consider may contain grids of unbounded sizes as subgraphs, the logic cannot axiomatize them. 

\begin{figure}
 \centering
 \includegraphics[scale=0.8]{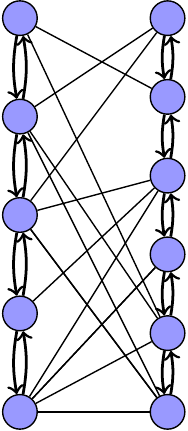}
   \ \ \ \ \ \ \ \ \ \ \ \ 
   \includegraphics[scale=0.8]{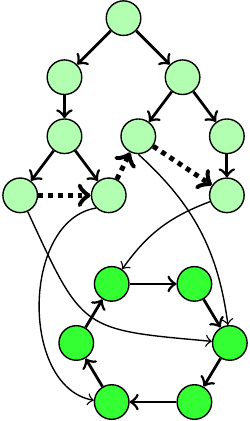}
   \\
   (a)  \ \ \ \ \ \ \  \ \ \ \ \ \ \ \ \ \ \  \ \ \ \ \ \ (b)
   
   \caption{\label{fig:expressive-graphs}}
 \end{figure}
 \begin{figure}
  \includegraphics[angle=90,scale=0.6]{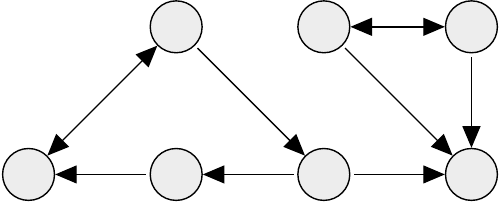}
   \ \ \ \ \ \ \ 
   \includegraphics[angle=90,scale=0.6]{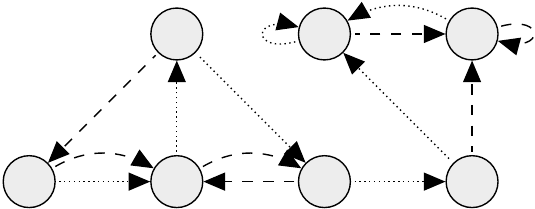}
   \\
   (a)  \ \ \ \ \ \ \  \ \ \ \ \  \ \ \ (b)
   
   \caption{\label{fig-oriented-new}   }
\end{figure}

\section{Background}\label{se:background}

This section introduces basic definitions and results in model theory and graph theory. We follow~\cite{ar:MakowskyTARSKI} and~\cite{bk:CourcelleEngelFriet12}. 

The {\em\bf two-variable fragment with counting} $C^2$ is the extension of the two-variable fragment of first order logic
with first order counting quantifiers $\exists^{\leq n}$, $\exists^{\geq n}$, $\exists^{=  n}$, for every $n\in\mathbb{N}$. 
Note that $C^2$ remains a fragment of first order logic. 
{\em\bf Monadic Second order logic $\MS$} 
is the extension of first order logic with set quantifiers which can range over elements of the universe
or subsets of relations\footnote{On relational structures, $\MS$ is also known as \emph{ Guarded Second Order logic $\mathrm{GSO}$}.
The results of this paper extend to $\CMS$, the extension of $\MS$ with modular counting quantifiers.}. 
Throughout the paper all structures consist of unary and binary relations only.
Structures are finite unless explicitly stated otherwise (in the discussion of $\WSone$). 
Let $\mcC$ be a vocabulary (signature). 
The {\em\bf arity} of a relation symbol $C\in \mcC$ is denoted by $\ar(C)$.
The set of unary (binary) relation symbols in $\mcC$ are {\em\bf $\unary(\mcC)$ ($\binary(\mcC)$)}. 
We write $\MS(\mcC)$ for the set of $\MS$-formulas on the vocabulary $\mcC$. 
The quantifier rank of a formula $\varphi\in \MS$, i.e.\ the maximal depth of nested quantifiers in $\varphi$ is denoted $\qr(\varphi)$. 
We denote by $\mfA_1 \sqcup\mfA_2$ the {\em\bf disjoint union} of two $\mcC$-structures $\mfA_1$ and $\mfA_2$. 
Given vocabularies $\mcC_1\subseteq \mcC_2$, a $\mcC_2$-structure $\mfA_2$ is an {\em\bf expansion} of a $\mcC_1$-structure $\mfA_1$ if $\mfA_1$ and $\mfA_2$ 
agree on the symbols in $\mcC_1$; in this case $\mfA_1$ is the {\em\bf reduct} of $\mfA_2$ to $\mcC_1$, i.e.\ $\mfA_1$ is 
the $\mcC_1$-reduct of $\mfA_2$. We denote the reduct of $\mfA_2$ to $\mcC_1$ by $\mfA_2|_{\mcC_1}$.
A $\mcC$-structure $\mfA_0$ with universe $A_0$ is a {\em\bf substructure} of a $\mcC$-structure $\mfA_1$ with universe $A_1$ 
if $A_0\subseteq A_1$ and for every $R\in\mcC$, $R^{\mfA_0} = R^{\mfA_1}\cap A_0^{\ar(R)}$. 
We say that $\mfA_0$ is the substructure of $\mfA_1$ generated by $A_0$.

{\em\bf Graphs} are structures of the vocabulary\footnote{Since we explicitly allowed quantification over subsets of relations for $\MS$, 
we do not view graphs and structures as incidence structures, in contrast to~\cite[Sections 1.8.1 and 1.9.1]{bk:CourcelleEngelFriet12}.  } 
$\mcC_G=\left\langle s \right\rangle$ consisting of a single binary relation symbol $s$. Graphs are simple and undirected
unless explicitly stated otherwise. 
{\em\bf Tree-width} $\tw(G)$ is a graph parameter indicating how close a simple undirected graph $G$ is to being a tree, cf.~\cite{bk:CourcelleEngelFriet12}. 
It is well-known that a graph has tree-width at most $k$ iff it is a partial $k$-tree. 
A {\em\bf partial $k$-tree} is a subgraph  of a $k$-tree.
{\em\bf $k$-trees} are built inductively from the $(k+1)$-clique by repeated addition of vertices, 
each of which is connected with $k$ edges to a $k$-clique. 
The {\em\bf  Gaifman graph $\gaif(\mfA)$} of a $\mcC$-structure $\mfA$ is
the graph whose vertex set is the universe of $\mfA$ and whose edge set is the union of 
the symmetric closures of $C^\mfA$ for every $C\in \binary(\mcC)$. 
Note the unary relations of $\mfA$ play no role in $\gaif(\mfA)$.
The {\em\bf tree-width $\tw(\mfA)$ of a $\mcC$-structure} $\mfA$ is the tree-width of its Gaifman graph. 
In this paper, tree-width is a parameter of \emph{finite} structures only. 
Fix $k\in\mathbb{N}$ for the rest of the paper. $k$ will denote the tree-width bound we consider. 

We introduce the notion of {\em\bf oriented $k$-trees} which refines the notion of $k$-trees. 
Let $\mcR = \{R_1,\ldots,R_k\}$ be a vocabulary consisting of binary relation symbols. 
An oriented $k$-tree is an $\mcR$-structure $\mfR$ in which all $R_i^\mfR$ are total functions
and whose Gaifman graph $\gaif(\mfR)$ is a partial $k$-tree. 
\begin{lem}\label{lem:oriented}
 Every $\mcC$-structure $\M$ of tree-width $k$ can be expanded into a $(\mcC\cup\mcR)$-structure $\mfN$
 such that:
 \begin{inparaenum}[(i)]
  \item $\mfN|_{\mcR}$ is an oriented $k$-tree,
  \item $\gaif(\mfN)$ is a subgraph of $\gaif(\mfN|_\mcR)$, and
  \item the tree-width of $\mfN$ is $k$. 
 \end{inparaenum}
\end{lem}

The oriented $2$-tree in Fig.~\ref{fig-oriented-new}(b) is an expansion 
 of the directed graph in Fig.~\ref{fig-oriented-new}(a)  as guaranteed in Lemma~\ref{lem:oriented}. 
 In Fig.~\ref{fig-oriented-new}(b), $R_1$ and $R_2$
 are denoted by the dashed arrows and the dotted arrows, respectively. 
 There are several other oriented $k$-trees which expand Fig.~\ref{fig-oriented-new}(a) and fulfill the requirements in Lemma~\ref{lem:oriented},

To see that Lemma~\ref{lem:oriented} holds, we describe a construction of $\mfN$ echoing the process of constructing $k$-trees above. 
For each vertex $u$ of the initial $(k+1)$-clique, we can set the values of $R_1^\mfN(u),\ldots,R_k^\mfN(u)$
to be the other $k$ vertices of the clique. When a new vertex $u$ is added to the $k$-tree, $k$ edges incident to it are added. 
We set $R_1^\mfN(u),\ldots,R_k^\mfN(u)$ to be the set of vertices incident to $u$. 
For oriented $k$-trees whose Gaifman graph is not a $k$-tree 
the construction of an oriented $k$-tree is augmented by changing the value of $R_i^\mfN(u)$ 
to $R_i^\mfN(u)=u$ whenever $R_i^\mfN(u)$ is not well-defined. This can happen when the target of $u$ under 
$R_i^\mfN$ is a vertex which was eliminated by taking the subgraph of a $k$-tree to obtain the partial $k$-tree. 

\section{Overview of the Main Theorem and its Proof}\label{se:overview}

The precise statement of the main theorem is as follows:
\begin{theorem}[Main Theorem]\label{th:main} Let $\mcC_{\bounded}$ and $\mcC_{\unbounded}$ be
vocabularies. Let $s$ be a binary relation symbol not in $\mcC_\bounded\cup \mcC_\unbounded$.
Let $\alpha\in\MS(\mcC_{\bounded})$ and $\beta\in C^{2}(\mcC_{\unbounded})$.
There is an effectively computable
 sentence $\delta\in C^{2}(\mcD)$ over a  vocabulary $\mcD\supseteq\{s\}$
such that the following are equivalent:
\begin{enumerate}[(i)]
 \item There is a $(\mcC_{\bounded}\cup\mcC_{\unbounded})$-structure
$\M$ such that $\M\models\alpha\land\beta$ and $\tw(\M|_{\mcC_{\bounded}})\leq k$.
\item There is a $\mcD$-structure
$\mfN$ such that $\mfN\models\delta$ and $s^{\mfN}$ is a binary tree.
\end{enumerate}
\end{theorem}

The first step towards proving Theorem~\ref{th:main} is the Separation Theorem:

\begin{theorem}[Separation Theorem]
\label{th:sep-proof-idea}
\label{th:sep}
Let $\mcC_{\bounded}$ and $\mcC_{\unbounded}$ be
vocabularies.
 Let $\formulaMSO \in \mso(\signatureMSO)$ and
$\formulaCTwo \in \tvlc(\signatureCTwo)$.
There are effectively computable sentences $\formulaMSOextended \in \mso(\signatureMSOextended)$ and
$\formulaCTwoExtended \in \tvlc(\signatureCTwoExtended)$
over vocabularies  $\signatureMSOextended$ and $\signatureCTwoExtended$ such that $\signatureMSOextended \cap \signatureCTwoExtended$ only contains unary relation symbols and
the following are equivalent:
\begin{enumerate}[(i)]
 \item There is a $(\signatureMSO \cup \signatureCTwo)$-structure $\model$ with $\model \models \formulaMSO \land \formulaCTwo$ and $\tw(\model|_\signatureMSO)\le k$
 \item There is a $(\signatureMSOextended \cup \signatureCTwoExtended)$-structure $\modelFinal$ with $\modelFinal \models \formulaMSOextended\land \formulaCTwoExtended$
 and $\tw(\modelFinal|_\signatureMSOextended)\le k$.
\end{enumerate}

\end{theorem}

In conjunction with Theorem~\ref{th:sep-proof-idea}, we only need to prove Theorem~\ref{th:main}
in the case that
the $\MS$-formula $\alpha$ and the $C^2$-formula $\beta$ only share unary relation symbols.
The significance of Theorem~\ref{th:sep-proof-idea} is that it allows us to use
tools designed for $\MS$ in our more involved setting. The proof
of Theorem~\ref{th:sep-proof-idea} uses notions of types for $C^2$-sentences
in Scott normal form, oriented $k$-trees, coloring arguments, and an induction
on ranks of structures. Theorem~\ref{th:sep-proof-idea} is discussed in Section~\ref{se:separation}.
The next step is to move from structures whose reducts have bounded tree-width
to structures which contain a binary tree.

\begin{lem}\label{lem:MSO-tree-interp-idea}
Let $\mcC_\bounded$ and $\mcC_\unbounded$ be vocabularies such that $\mcC_\bounded\cap\mcC_\unbounded$ contains only
unary relation symbols.
Let $s$ be a binary relation symbol.
There is a vocabulary $\mcD_\bounded$ consisting of $s$ and unary relation symbols only,
and,
for every $\alpha \in \MS(\mcC_\bounded)$ and $\beta\in C^2(\mcC_\unbounded)$,
effectively computable
sentences $\alpha' \in \MS(\mcD_\bounded)$ and $\beta'\in C^2(\mcD_\bounded\cup\mcC_\unbounded)$
such that
 the following are equivalent:
\begin{enumerate}[(i)]
\item There is a $(\mcC_\bounded\cup \mcC_\unbounded)$-structure
$\M$ such that $\M\models\alpha\land \beta$ and such that $\tw(\M|_{\mcC_\bounded})\leq k$.
\item There is a $(\mcD_\bounded\cup \mcC_\unbounded)$-structure
$\mfN$ such that $\mfN\models\alpha'\land\beta'$ and $s^{\mfN}$ is a binary tree.
\end{enumerate}
\end{lem}

Technically, Lemma~\ref{lem:MSO-tree-interp-idea} is proved using
a translation scheme which maps structures with a binary tree into
structures whose $\mcC_\bounded$-reducts have tree-width at most $k$, and conversely,
each of the latter structures is the image of a structure with a binary tree under the translation scheme.
Translation schemes capturing the graphs of tree-width at most $k$ as the image of
labeled trees were studied in the context of
decidability and model checking of $\MS$~\cite{ar:ArnborgEtAl}.
We need a more refined construction to ensure that the translation scheme
also behaves correctly on $C^2$-sentences, i.e.\ that it maps $C^2$-sentences to $C^2$-sentences,
see Lemma~\ref{lem:MSO-tree-interp-idea} in Section~\ref{se:translations}.

Now that we have reduced our attention to the case that
our structures contain a binary tree,
we can replace $\MS$-sentences with equi-satisfiable $C^2$-sentences.

\begin{lem}\label{lem:HintikkaTree-idea}
Let $\mcC$
be a vocabulary which consists only of a binary relation symbol $s$
and unary relation symbols.
Let $\alpha$ be an $\MS(\mcC)$-sentence.
There is an effectively computable $C^2(\mcD)$-sentence $\gamma$
over a vocabulary $\mcD \supseteq \mcC$
such that for every $\mcC$-structure $\M$ in which $s^\M$ is a binary tree
the following are equivalent:\\[5pt]
\begin{inparaenum}[(i)]
 \item $\M\models\alpha$
\item There is a $\mcD$-structure
$\mfN$ expanding $\M$ such that $\mfN\models\gamma$.
\end{inparaenum}
\end{lem}

For the proof of Lemma~\ref{lem:HintikkaTree-idea}
we use a Feferman-Vaught type theorem which states that the Hintikka type (i.e. $\MS$ types)
of a binary tree labeled with unary relation symbols depends only
on the Hintikka types of its children. We can therefore axiomatize in $C^2$
that the Hintikka type of the labeled binary tree implies a given $\MS$-sentence.

Having replaced the $\MS$-sentence in statement (ii) of Lemma~\ref{lem:MSO-tree-interp-idea} with a $C^2$-sentence, we are left with the problem
of deciding whether a $C^2$-sentence is satisfiable by a structure in which a specified relation is a binary tree,
which has recently been shown to be decidable:

\begin{theorem}[Charatonik and Witkowski~\cite{pr:CW13}]
Let $\mcC$ be a vocabulary which contains a binary relation symbol $s$.
Given a $C^2(\mcC)$-sentence $\varphi$, it is decidable whether $\varphi$
is satisfiable by a structure $\M$ in which $s^\M$ is a binary tree.
\end{theorem}

%
%

\section{Separation Theorem}\label{se:separation}

\subsection{Basic Definitions and Results}

\subsection*{$1$-types and $2$-types}
We begin with some notation and definitions in the spirit of the literature
on decidability of $\tvlc$, cf.\ e.g.~\cite{ar:PH05,pr:CW13}.
Let $\signatureDefs$ be a vocabulary of unary and binary relations.

A \emph{\bf 1-type} $\oneType$ is a maximal consistent set of atomic $\signatureDefs$-formulas or negations of atomic $\signatureDefs$-formulas with free variable $x$, i.e.,
exactly one of  $A(x)$ and $\neg A(x)$ belongs to $\oneType$ for every unary relation symbol $A \in \signatureDefs$,
and exactly one of  $B(x,x)$ and $\neg B(x,x)$ belongs to $\oneType$
for every binary relation symbol $B \in \signatureDefs$.
We denote by $\typeFormula_\oneType(x) = \bigwedge_{\typeFormulaConjunct \in \oneType} \typeFormulaConjunct$ the formula that characterizes the 1-type $\oneType$.
We denote by $\oneTypes(\signatureDefs)$ the set of 1-types over $\signatureDefs$.

A \emph{\bf 2-type} $\twoType$ is a maximal consistent set of atomic $\signatureDefs$-formulas or
negations of atomic $\signatureDefs$-formulas with free variables $x$ and $y$,
i.e., for every $z \in \{x,y\}$ and unary relation symbol $A \in \signatureDefs$, exactly one of $A(z)$ and $\neg A(z)$ belongs to $\twoType$, and for every $z_1,z_2 \in \{x,y\}$ and binary relation symbol $B \in \signatureDefs$, exactly one of $B(z_1,z_2)$ and $\neg B(z_1,z_2)$ belongs to $\twoType$.
We note that the equality relation $\approx$ is also part of a $2$-type.
We write $\twoType^{-1}$ for the 2-type obtained from $\twoType$ by substituting all occurrences of $x$ resp. $y$ with $y$ resp. $x$.
We write $\twoType_x$ for the 1-type obtained from $\twoType$ by restricting $\twoType$ to formulas with free variable $x$.
We write $\twoType_y$ for the 1-type obtained from $\twoType$ by restricting $\twoType$ to formulas with free variable $y$ {\em and substituting $y$ with $x$}.
We denote by $\typeFormula_\twoType(x,y) = \bigwedge_{\typeFormulaConjunct \in \twoType} \typeFormulaConjunct$ the formula that characterizes the 2-type $\twoType$.
We denote by $\twoTypes(\signatureDefs)$ the set of 2-types over $\signatureDefs$.

Let $\model$ be a $\signatureDefs$-structure.
We denote by $\oneTypeOf^\model(u)$ the unique 1-type $\oneType$ such that $\model \models \typeFormula_\oneType(u)$.
For elements $u,v$ of $\model$, we denote by $\twoTypeOf^\model(u,v)$ the unique 2-type $\twoType$ such that $\model \models \typeFormula_\twoType(u,v)$.
We denote by $\twoTypeOf(\model) = \{ \twoTypeOf^\model(u,v) \mid u,v \text{ elements of } \model \}$ the set of 2-types \emph{realized} by $\model$.
The following lemma is easy to see:

\begin{lem}
\label{lem:star-type-equivalence}
Let $\model_1, \model_2$ be two $\signatureDefs$-structures over the same universe $\modelUniverse$ and let $\phi = \forall x,y.\, \chi \in \tvlc(\signatureDefs)$ with $\chi$ quantifier-free.
If $\twoTypeOf(\model_1) = \twoTypeOf(\model_2)$, then $\model_1 \models \phi$ iff $\model_2 \models \phi$.\\
(See proof in Appendix~\ref{ap:lem:star-type-equivalence}.)
\end{lem}

\subsection*{Scott Normal Form and $\messageAlphabet$-functionality}

$\tvlc$-sentences have a Scott-Normal Form, cf.~\cite{conf/stacs/GradelOR97}, which can be obtained by iteratively applying Skolemization and introducing new predicates for subformulas, together with predicates ensuring the soundness of this transformation:

\begin{lem}[Scott Normal Form,~\cite{conf/stacs/GradelOR97}]
For every $\tvlc$-sentence $\formulaCTwo$ there is a $\tvlc$-sentence $\formulaCTwo'$ of the form
\begin{equation}\label{eq:normal-form}
\forall x,y.\,\chi \wedge \bigwedge_{i \in [\messageTypeNumber]} \forall x.\, \exists^{=1} y.\, \messageTypeCTwo_i(x,y),
\end{equation}
with $\chi$ quantifier-free, over an expanded vocabulary
such that $\formulaCTwo$ and $\formulaCTwo'$ are equi-satisfiable.
Moreover, $\formulaCTwo'$ is computable.
The expanded vocabulary contains in particular
the fresh binary relation symbols in $\messageAlphabetCTwo = \{\messageTypeCTwo_1, \ldots, \messageTypeCTwo_\messageTypeNumber\}$

\end{lem}

Let $\messageAlphabet$ be a set of binary relation symbols.
We say a structure $\model$ is \emph{\bf $\messageAlphabet$-functional}, if for every $\messageType \in \messageAlphabet$,
$\messageType^\model$ is a total function on the universe of $\model$.
Observe the following are equivalent for every structure $\model$:\\[5pt]
\begin{inparaenum}[(i)]
 \item $\model$ satisfies Eq. (\ref{eq:normal-form}), and 
 \item $\model \models \forall x,y.\,\chi$
and $\model$ is $\messageAlphabetCTwo$-functional.

\end{inparaenum}

\subsection*{Message Types and Chromaticity}

Let $\messageAlphabet \subseteq \binary(\signatureDefs)$ be a subset of the binary relation symbols of $\signatureDefs$.
We write $\twoType \in \MsgTypes{\messageAlphabet}(\signatureDefs)$ and say $\twoType$ is a \emph{\bf $\messageAlphabet$-message type}, if $\twoType \in \twoTypes(\signatureDefs)$ and $\messageType(x,y) \in \twoType$ for some $\messageType \in \messageAlphabet$.
Let $\model$ be a $\signatureDefs$-structure with universe $\modelUniverse$.
We define $E^\model_\messageAlphabet = \{ (u,v) \in \modelUniverse^2 \mid \text{there is a } \messageType \in \messageAlphabet \text{ with } \model \models \messageType(u,v) \}$.
The \emph{\bf $\messageAlphabet$-message-graph} is the directed graph $\messageGraph^\model_\messageAlphabet = (\modelUniverse,E)$,
where $E = \{ (u,v) \in \modelUniverse^2 \mid u \neq v \text{ and } (u,v) \in E^\model_\messageAlphabet \circ E^\model_\messageAlphabet \}$,
where $R \circ S = \{ (a,b) \mid \text{ there is a } c \text{ with } (a,c) \in R \text{ and } (c,b) \in S \}$ denotes the usual \emph{composition} of relations.
We say $\model$ is \emph{\bf $\messageAlphabet$-chromatic}, if $\oneTypeOf^\model(u) \neq \oneTypeOf^\model(v)$ for all $(u,v) \in E$.


We note that if $\model$ is $\messageAlphabet$-functional, then $\messageGraph^\model_\messageAlphabet$ has out-degree $\outDegree(u) \le |\messageAlphabet|^2$ for all $u \in \modelUniverse$.
This allows us to prove Lemma~\ref{lem:chromatic-models} based on Lemma~\ref{lem:coloring}; the proofs can be found in Appendices~\ref{ap:lem:coloring} and~\ref{ap:lem:chromatic-models}.

\begin{lem}
\label{lem:chromatic-models}
There is a finite set of unary relations symbols $\colors_\messageAlphabet$ such that every $\messageAlphabet$-functional 
$\signatureDefs$-structure can be expanded to a $\messageAlphabet$-chromatic $(\signatureDefs \cup \colors_\messageAlphabet)$-structure.
\end{lem}

\begin{lem}
\label{lem:coloring}
Let $G = (V,E)$ be a directed graph with out-degree $\outDegree(v) \le k$ for all $v \in V$.
Then, the underlying undirected graph has a proper $2k+1$-coloring.
\end{lem}


\subsection{Separation Theorem}

Let $\boundedGraph = (V,E)$ be an (undirected) graph.
We say $\boundedGraph$ is \emph{\bf $k$-bounded}, if the edges of $\boundedGraph$ can be oriented such that every node of $\boundedGraph$ has out-degree less than $k$.
We say a structure $\model$ is $k$-bounded if its Gaifman graph is $k$-bounded.

\begin{septhm}[Separation Theorem] 
Let $k$ be a natural number.
Let $\mcC_{\bounded}$ and $\mcC_{\unbounded}$ be vocabularies.
Let $\formulaMSO \in \mso(\signatureMSO)$ and\footnote{The Separation Theorem remains correct if we replace 
$\tvlc$ with any logic containing $\tvlc$ which is closed under conjunction. 
}
$\formulaCTwo \in \tvlc(\signatureCTwo)$.
There are effectively computable sentences $\formulaMSOextended \in \mso(\signatureMSOextended)$ and
$\formulaCTwoExtended \in \tvlc(\signatureCTwoExtended)$
over vocabularies  $\signatureMSOextended$ and $\signatureCTwoExtended$ such that
$\signatureMSOextended \cap \signatureCTwoExtended$ only contains unary relation symbols such that
for every $k$-bounded graph $\boundedGraph$ the following are equivalent:
\begin{enumerate}[(i)]
 \item There is a $(\signatureMSO \cup \signatureCTwo)$-structure $\model$ with $\model \models \formulaMSO \land \formulaCTwo$ and $\gaif(\model|_\signatureMSO) = \boundedGraph$.
 \item There is a $(\signatureMSOextended \cup \signatureCTwoExtended)$-structure $\modelFinal$ with $\modelFinal \models \formulaMSOextended\land \formulaCTwoExtended$
 and $\gaif(\modelFinal|_\signatureMSOextended) = \boundedGraph$.
\end{enumerate}

\end{septhm}

We assume that $\formulaCTwo$ is in the form given in Eq.~(\ref{eq:normal-form}) for some set of binary relation symbols $\messageAlphabetCTwo = \{\messageTypeCTwo_1, \ldots, \messageTypeCTwo_\messageTypeNumber\} \subseteq \signatureCTwo$ and quantifier-free $\tvlc$-formula $\chi$.
Let $\signatureOrientedKTree = \{\gaifmannRelation_1, \ldots, \gaifmannRelation_k\}$ be a set of fresh binary relation symbols.
We set $\messageAlphabet = \messageAlphabetCTwo \cup \signatureOrientedKTree$.
We begin by giving an intuition for the proof of the Separation Theorem in three stages.

\subsection*{1. Syntactic separation coupled with semantic constraints}

For a binary relation symbol $\intersectionRelation$, we define its \emph{copy} as the relation symbol $\intersectionRelationCopy$.
For every vocabulary $\signatureDefs$, we define its \emph{copy} $\copySymbol{\signatureDefs} = \unary(\signatureDefs) \cup \{ \intersectionRelationCopy \mid \intersectionRelation \in \signatureDefs \}$ to be the unary relation symbols of $\signatureDefs$ plus the copies of its binary relations symbols.
We assume that copied relation symbols are distinct from non-copied symbols, i.e., $\binary(\signatureDefs) \cap \binary(\copySymbol{\signatureDefs}) = \emptyset$.
For a formula $\varphi$ over vocabulary $\signatureDefs$, we define its \emph{copy} $\copySymbol{\varphi}$ over vocabulary $\copySymbol{\signatureDefs}$ as the formula
obtained from $\varphi$ by substituting
every occurrence of a binary relation symbol $\intersectionRelation \in \signatureDefs$ with $\intersectionRelationCopy$.

The sentences $\copySymbol{\formulaMSO}$ (the copy of $\formulaMSO$) and $\formulaCTwo$  do not share any binary relation symbols.
Clearly, \textit{(i)} from Theorem~\ref{th:sep-proof-idea} holds iff
\begin{enumerate}[\textit{(I)}]
\item\label{enum:intutionI} $\copySymbol{\formulaMSO} \land \formulaCTwo$ is satisfied by a
$((\copySymbol{\signatureMSO} \cup \copySymbol{\signatureCTwo}) \cup (\signatureMSO \cup \signatureCTwo))$-structure $\modelFinal$ with
$\intersectionRelation^\modelFinal = \intersectionRelationCopy^\modelFinal$
for all $\intersectionRelation \in \binary(\signatureMSO)$ and $\gaif(\model|_\signatureMSO) = \boundedGraph$.
\end{enumerate}

\noindent
In the next two stages we will construct $\formulaMSOextended$ and $\formulaCTwoExtended$
so
that \textit{(I)} is equivalent to
 \textit{(ii)} from Theorem~\ref{th:sep-proof-idea}.
 More precisely, we will
construct sentences
$\adjunctFormula_\bounded,\adjunctFormula_\unbounded \in \tvlc(\signatureCTwoExtended)$ with
$\signatureCTwoExtended \supseteq \signatureMSO \cup \signatureCTwo$ and
$\signatureMSOextended = \copySymbol{\signatureCTwoExtended}$ such that \textit{(I)} is equivalent
\textit{(II)}:

\begin{enumerate}[\textit{(II)}]
\item\label{enum:intutionII} $(\copySymbol{\formulaMSO} \land \copySymbol{\adjunctFormula_\bounded}) \land (\formulaCTwo \land \adjunctFormula_\unbounded)$ is satisfied by a
$(\signatureMSOextended \cup \signatureCTwoExtended)$-structure $\modelFinal$ with  $\gaif(\modelFinal|_\signatureMSOextended) = \boundedGraph$.
\end{enumerate}


\subsection*{2. Representation of $k$-bounded structures using functions and unary relations.}

Theorem~\ref{th:sep-proof-idea} as well as \textit{(I)} and \textit{(II)} involve reducts which are $k$-bounded structures.
$k$-bounded $\signatureDefs$-structures $\modelA$
can be represented by introducing new binary relation symbols interpreted as functions and new unary relation symbols as follows.
\begin{enumerate}[(a)]
 \item We add $k$ fresh relation symbols $\signatureOrientedKTree = \{\gaifmannRelation_1, \ldots, \gaifmannRelation_k\}$
and axiomatize that these relations are interpreted as total functions.
\item We add fresh unary relations $\{P_\twoType \mid \twoType \in  \MsgTypes{\signatureOrientedKTree}(\signatureDefs)\}$
and axiomatize that every element labeled by $P_\twoType$ has an outgoing edge with 2-type $\twoType$.
The symbols $P_\twoType$ are called \emph{unary $2$-type annotations}.
\item We axiomatize that $\gaif(\modelA|_{\binary(\signatureDefs)})=\gaif(\modelA|_\signatureOrientedKTree)$.
\end{enumerate}

In other words, the functions interpreting $\gaifmannRelation_1, \ldots, \gaifmannRelation_k$ witness that $\modelA$
can be oriented so that every node in the Gaifman graph of $\modelA$ has outdegree at most $k$. The $2$-type of each edge $(u,v)$ in
$\modelA$ is encode by putting the unary relation symbol $P_\twoType$ of the $2$-type of $(u,v)$ on the source $u$
in the orientation. 

Given a $((\copySymbol{\signatureMSO} \cup \copySymbol{\signatureCTwo}) \cup (\signatureMSO \cup \signatureCTwo))$-structure $\modelFinal$, we will use the above representation twice, on $\modelFinal|_\signatureMSO$ and $\modelFinal|_{\copySymbol{\signatureMSO}}$, 
by axiomatizing that every element labeled by $P_\twoType$ has an outgoing edge with 2-type $\twoType$ and an outgoing edges with 2-type $\copySymbol{\twoType}$.
This will allow us to replace the condition from \textit{(I)} that $\intersectionRelation^\modelFinal = \intersectionRelationCopy^\modelFinal$ for all
$\intersectionRelation \in \binary(\signatureMSO)$ with the condition that
$\gaifmannRelation_i^\modelFinal = \copySymbol{\gaifmannRelation_i}^\modelFinal$ for all $\gaifmannRelation_i \in \signatureOrientedKTree$.
We will define a vocabulary 
$\signatureMain \supseteq \signatureMSO \cup \signatureCTwo \cup \signatureOrientedKTree $.
Let $\signatureP_{\signatureOrientedKTree} = \{P_\twoType \mid \twoType \in  \MsgTypes{\signatureOrientedKTree}(\signatureMain)\}$. 
We will define two vocabularies 
$\signatureCTwoExtended \supseteq \signatureMSO \cup \signatureCTwo
\cup \signatureOrientedKTree \cup \signatureP_{\signatureOrientedKTree}$ and $\signatureMSOextended = \copySymbol{\signatureCTwoExtended}$. 
According to (a), (b), and (c), we will construct
$\nu_\bounded,\nu_\unbounded \in \tvlc(\signatureCTwoExtended)$ such that
\textit{(I)} is equivalent to the following:
\begin{enumerate}[\textit{(I')} ]
\item $(\copySymbol{\formulaMSO}\land \copySymbol{\nu_\bounded} ) \land (\formulaCTwo \land \nu_\unbounded)$ is satisfied by a
$\signatureMSOextended\cup\signatureCTwoExtended$-structure $\model$ with
$\gaifmannRelation_i^\model = \copySymbol{\gaifmannRelation_i}^\model$
for all $\gaifmannRelation_i \in \signatureOrientedKTree$
 and $\gaif(\model|_\signatureOrientedKTree) = \boundedGraph$.
\end{enumerate}

%

%

\subsection*{3. Establishing the semantic condition of \textit{(I')} by swapping edges.}
Here we discuss how to show the implication from \textit{(II)} to \textit{(I')} and make the vocabularies
$\signatureMSOextended$ and $\signatureCTwoExtended$ precise.
Let $\modelFinal$ be a $(\signatureMSOextended \cup \signatureCTwoExtended )$-structure
with $\modelFinal \models (\copySymbol{\formulaMSO} \land \copySymbol{\adjunctFormula_\bounded}) \land (\formulaCTwo \land \adjunctFormula_\unbounded)$.
It simplifies the discussion to split a
$(\signatureMSOextended \cup \signatureCTwoExtended )$-structure $\modelFinal$ into two
$\signatureCTwoExtended$-structures.
The $\signatureCTwoExtended$-structure $\modelAlt$ is $ \modelFinal|_\signatureCTwoExtended$.
The $\signatureCTwoExtended$-structure $\modelAlt'$ is obtained from $\modelFinal|_\signatureMSOextended$ by renaming
copies of relation symbols $\copySymbol{B}$ to $B$  --- i.e., we define the $\signatureCTwoExtended$-structure $\modelAlt'$ by setting $\oneTypeOf^{\modelAlt'}(u) = \oneTypeOf^\modelFinal(u)$ for
all $u \in \modelUniverse$ and setting $\modelAlt' \models \intersectionRelation(u,v)$ iff $\modelFinal \models \intersectionRelationCopy(u,v)$ for all $u,v \in \modelUniverse$ and $\intersectionRelation \in \binary(\signatureMSOextended)$.
The interpretations of the relations $\gaifmannRelation_i$ might differ in $\modelAlt$ and $\modelAlt'$.
Observe that we have $\modelAlt'\models \formulaMSO$ and $\modelAlt \models \formulaCTwo$.
The \emph{key idea of the proof} is prove the existence of a sequence of structures $\modelAlt = \modelAlt_0, \ldots, \modelAlt_p$,
where each $\modelFinal_{i+1}$ is obtained from $\modelFinal_i$  by \emph{swapping edges},
until the interpretations of the relations $\gaifmannRelation_i$ agree in $\modelAlt_p$ and $\modelAlt'$.
The edge swapping operation is a local operation which involves changing the $2$-types of at most $4$ edges.

The edge swapping operation satisfies two crucial preservation requirements: edge swapping preserves (PR-1) the truth value of $\formulaCTwo$, i.e.
$\modelAlt_p \models \formulaCTwo$, and (PR-2) $\signatureOrientedKTree$-functionality.
The universal constraint $\forall x,y.\,\chi$ in $\formulaCTwo$
is maintained under edge swapping because this operation does not change the set of $2$-types
(see Lemma~\ref{lem:star-type-equivalence}). To satisfy the preservation requirements (PR-1) and (PR-2), all that remains is to guarantee
the existence of a sequence of edge swapping preserving
$\messageAlphabet$-functionality. 
Note that $\messageAlphabet$-functionality amounts to $\signatureOrientedKTree$-functionality and  $\messageAlphabetCTwo$-functionality.
We use two main techqniues for ensuring the  preservation of $\messageAlphabet$-functionality:
chromaticity and unary $2$-type annotations.
We will axiomatize that the structures $\modelAlt$ and $\modelAlt'$ are chromatic and we will take care that chromaticity is maintained during edge swaps.
We will add fresh unary relation symbols
 $\signatureP_{\messageAlphabetCTwo} = \{P_\twoType \mid \twoType \in  \MsgTypes{\messageAlphabetCTwo}(\signatureMain)\}$
 and axiomatize that every element of $\modelFinal$ labeled by $P_\twoType$ has an outgoing edge with 2-type $\twoType$ and an outgoing edge with 2-type $\copySymbol{\twoType}$.

\subsection*{Proof of the Separation Theorem}

\vspace{4pt}
We now start the formal proof of the Separation Theorem.
Let $\colors_\messageAlphabet$ be the vocabulary from Lemma~\ref{lem:chromatic-models}.
We set $\signatureMain = \signatureMSO \cup \signatureCTwo \cup \signatureOrientedKTree \cup \colors_\messageAlphabet$.
We set $\signatureMSOextended = \signatureMain \cup \signatureP_{\messageAlphabetCTwo} \cup \signatureP_{\signatureOrientedKTree}$ 
and $\signatureCTwoExtended = \copySymbol{\signatureMSOextended}$.
Next we will define formulas $\formulaMSOcopy \in \mso(\signatureMSOextended)$ and $\formulaCTwoExtended \in \tvlc(\signatureMSOextended)$, and set $\formulaMSOextended = \copySymbol{\formulaMSOcopy} \in \mso(\signatureCTwoExtended)$.
We set $\nu_\bounded = \functionFormula \land \gaifmannFormula \land \containmentFormula \land \definitionFormula_\bounded$, 
$\nu_\unbounded = \functionFormula \land \gaifmannFormula \land \definitionFormula_\unbounded$, 
$\adjunctFormula_\bounded = \nu_\bounded \land \colorFormula$, and 
$\adjunctFormula_\unbounded  = \nu_\unbounded \land \colorFormula$, 
$\formulaMSOcopy = \formulaMSO \land \adjunctFormula_\bounded$, 
and $\formulaCTwoExtended = \formulaCTwo \land \adjunctFormula_\unbounded$, where:

\[
\begin{array}{lll}
 \functionFormula & = & \bigwedge_{i \in [k]} \forall x.\, \exists^{= 1} y.\, \gaifmannRelation_i(x,y) \\[7pt]
 \gaifmannFormula & = & \forall x,y.\,
    x \not\approx y \rightarrow \left(
    \bigvee_{i \in [k]} \gaifmannRelation_i(x,y) \vee \gaifmannRelation_i(y,x) 
    \leftrightarrow \bigvee_{\intersectionRelation \in \binary(\signatureMSO)} \intersectionRelation(x,y) \vee \intersectionRelation(y,x) \right)\\
 \containmentFormula & = & \bigwedge_{\intersectionRelation \in \binary(\signatureCTwo)} \forall x,y.\,
    x \not\approx y \rightarrow \left(
    \intersectionRelation(x,y) \rightarrow \bigvee_{i \in [k]} \gaifmannRelation_i(x,y) \vee \gaifmannRelation_i(y,x) \right)\\
 \end{array}
 \]
 
 \[
 \begin{array}{lll}
 \definitionFormula_\bounded & = & \displaystyle
 \bigwedge_{
           \scriptsize
           \begin{aligned}
            &\,\,\twoType \in \MsgTypes{\signatureOrientedKTree}(\signatureMain) \mbox{ or } \\[-4pt]
            &(\twoType \in \MsgTypes{\messageAlphabet}(\signatureMain) \mbox{ and }\twoType^{-1} \in \MsgTypes{\signatureOrientedKTree}(\signatureMain))\\[-5pt]
            &
            \end{aligned}
           }
 \forall x.\, P_\twoType(x) \leftrightarrow \exists y. \typeFormula_\twoType(x,y)\\[7pt]
 \definitionFormula_\unbounded & = & \displaystyle \bigwedge_{
           \scriptsize
           \begin{aligned}
            &\twoType \in \MsgTypes{\messageAlphabet}(\signatureMain)\\[-5pt]
            &
            \end{aligned}
            } \forall x.\, P_\twoType(x) \leftrightarrow \exists y. \typeFormula_\twoType(x,y)\\[7pt]
 \colorFormula & = &\displaystyle \bigwedge_{\scriptsize
                                              \begin{aligned}
                                                 & \twoType \in \MsgTypes{\messageAlphabet}(\signatureMain) \\[-4pt] 
                                                 & \messageType, \messageType' \in \messageAlphabet:\,\messageType'(x,y) \in \twoType
                                               \end{aligned} 
                                            }
 \forall x,y.\, P_\twoType(x) \land \messageType(y,x) \land \neg \messageType'(x,y) \rightarrow \neg \typeFormula_{\twoType_y}(y)
 \end{array}
\]

$\functionFormula$ expresses that each $\gaifmannRelation_i$ is interpreted as a total function,
$\gaifmannFormula$ expresses that for every $\signatureCTwoExtended$-structure $\modelAlt$ with 
$\modelAlt \models \gaifmannFormula$ we have that $\gaif(\modelAlt|_\signatureOrientedKTree) = \gaif(\modelAlt|_\signatureMSO)$,
$\containmentFormula$ expresses that for every $\signatureCTwoExtended$-structure $\modelAlt$ with 
$\modelAlt \models \containmentFormula$ we have that $\gaif(\modelAlt|_\signatureCTwoExtended)$ is a subgraph of $\gaif(\modelAlt|_\signatureOrientedKTree)$,
$\definitionFormula$ expresses that for every $\twoType \in \MsgTypes{\messageAlphabet}(\signatureMain)$ with 
$\twoType \in \MsgTypes{\signatureOrientedKTree}(\signatureMain)$ or 
$\twoType \in \MsgTypes{\signatureOrientedKTree}(\signatureMain)$, 
$\signatureCTwoExtended$-structure $\modelAlt$ and element $u$ of $\modelAlt$ we have $\modelAlt \models P_\twoType(u)$ 
iff there is an element $v$ of $\modelAlt$ such that $\twoTypeOf^{\modelAlt|_\signatureMain}(u,v) = \twoType$,
$\definitionFormula$ expresses that for every $\twoType \in \MsgTypes{\messageAlphabet}(\signatureMain)$, $\signatureCTwoExtended$-structure 
$\modelAlt$ and element $u$ of $\modelAlt$ we have $\modelAlt \models P_\twoType(u)$ iff there is an element $v$ of $\modelAlt$ such that $\twoTypeOf^{\modelAlt|_\signatureMain}(u,v) = \twoType$,
$\colorFormula$ expresses that for every $\signatureCTwoExtended$-structure $\modelAlt$ with $\modelAlt \models \definitionFormula_\unbounded$, $\modelAlt$ is chromatic iff $\modelAlt \models \colorFormula$.

The direction ``(i) implies (ii)'' of the Separation Theorem is not hard:

\begin{lem}\label{lem:expansion}
Let $\boundedGraph$ be a $k$-bounded graph.
Let $\model$ be a $(\signatureMSO \cup \signatureCTwo)$-structure with $\model \models \formulaMSO \wedge \formulaCTwo$ and $\gaif(\model|_\signatureMSO) = \boundedGraph$.
$\model$ can be expanded to a $(\signatureMSOextended \cup \signatureCTwoExtended)$-structure $\modelFinal$ with
$\modelFinal \models \formulaMSOextended \wedge \formulaCTwoExtended$ and $\gaif(\modelFinal|_\signatureMSOextended) = \boundedGraph$.
\end{lem}
\begin{proof}
Because of $\gaif(\model|_\signatureMSO) = \boundedGraph$ and $G$ is $k$-bounded, we can expand $\model$ to a $\signatureMSO \cup \signatureCTwo \cup \signatureOrientedKTree$-structure $\modelHat$ such that $\gaif(\modelAlt|_\signatureOrientedKTree) = \gaif(\modelAlt|_\signatureMSO)$ and
$\gaifmannRelation^\modelHat_i$ is a total function for all $i \in [k]$ (possibly adding self-loops for the relations $\gaifmannRelation_i$).
Thus, $\modelHat \models \functionFormula \wedge \gaifmannFormula$.
According to Lemma~\ref{lem:chromatic-models}, $\modelHat$ can be expanded to a chromatic structure $\modelTilde$ over vocabulary $\signatureMain$ with $\modelTilde \models \colorFormula$.
We expand $\modelTilde$ to a $\signatureMSOextended$-structure $\modelAlt$ such that for all $u \in \modelUniverse$ and $\twoType \in \MsgTypes{\messageAlphabet}(\signatureMain)$ we have $\modelAlt \models P_\twoType(u)$ iff  there is an element $v$ of $\modelAlt$ such that $\twoTypeOf^{\modelAlt|_\signatureMain}(u,v) = \twoType$.
This definition gives us $\modelAlt \models \definitionFormula_\unbounded$, and thus $\modelAlt \models \formulaCTwoExtended$.
We expand $\modelAlt$ to a $(\signatureMSOextended \cup \signatureCTwoExtended)$-structure $\modelFinal$ such that for all $u,v \in \modelUniverse$ and $\intersectionRelation \in \signatureIntersection$ we have
$\modelFinal \models \intersectionRelationCopy(u,v)$ iff
$\modelFinal \models \intersectionRelation(u,v)$ and $\modelFinal \models \gaifmannRelation_i(u,v)$ or $\modelFinal \models \gaifmannRelation_i(v,u)$ for some $i \in [k]$.
We note that $\modelFinal \models \formulaMSOextended$.
\end{proof}

Now we turn to the direction ``(ii) implies (i)''.
Let $\boundedGraph$ be a $k$-bounded graph.
Let $\modelFinal$ be a $(\signatureMSOextended \cup \signatureCTwoExtended )$-structure with $\modelFinal \models \formulaMSOextended \wedge \formulaCTwoExtended$ and
$\gaif(\modelFinal|_\signatureMSOextended) = \boundedGraph$.
Let $\modelUniverse$ be the universe of $\modelFinal$.
We define the $\signatureCTwoExtended$-structure $\modelAlt'$ by setting $\oneTypeOf^{\modelAlt'}(u) = \oneTypeOf^\modelFinal(u)$ for all $u \in \modelUniverse$ and setting $\modelAlt' \models \intersectionRelation(u,v)$ iff $\modelFinal \models \intersectionRelationCopy(u,v)$ for all $u,v \in \modelUniverse$ and $\intersectionRelation \in \binary(\signatureMSOextended)$.
We note that $\modelAlt' \models \formulaMSOcopy$ and $\gaif(\modelAlt') = G$.
We define the $\signatureCTwoExtended$-structure $\modelAlt$ by setting $\modelAlt = \modelFinal|_\signatureCTwoExtended$.
We note that $\modelAlt \models \formulaCTwoExtended$.

We make the following definition:
For $u \in \modelUniverse$ and $i \in [k]$ we set $\rank^i_u(\modelAlt,\modelAlt') = 1$, if there are $v,w \in \modelUniverse$ with $\modelAlt \models \gaifmannRelation_i(u,v)$, $\modelAlt' \models \gaifmannRelation_i(u,w)$ and $v \neq w$; we set $\rank^i_u(\modelAlt,\modelAlt') = 0$, otherwise.
We set $\rank_u(\modelAlt,\modelAlt')  = \sum_{i \in [k]} \rank^i_u(\modelAlt,\modelAlt')$ and $\rank(\modelAlt,\modelAlt') = \sum_{u \in \modelUniverse} \rank_u(\modelAlt,\modelAlt')$.
$\rank$ measures the deviation of the relations $\signatureOrientedKTree$ in $\modelAlt$ and $\modelAlt'$ (we note that there always are unique $v,w \in \modelUniverse$ for $u \in \modelUniverse$ with $\modelAlt \models \gaifmannRelation_i(u,v)$, $\modelAlt' \models \gaifmannRelation_i(u,w)$ because of $\modelAlt \models \functionFormula$ and $\modelAlt' \models \functionFormula$) and has the following important property:

\begin{lem}\label{lem:no-deviation}
If $\rank(\modelAlt,\modelAlt') = 0$, then $\modelAlt|_{\signatureMSO \cup \signatureCTwo} \models \formulaMSO \wedge \formulaCTwo$ and $\gaif(\modelAlt_\signatureMSO) = \boundedGraph$.
\\
(See Appendix~\ref{ap:lem:no-deviation} for the proof.)
\end{lem}
The proof of Lemma~\ref{lem:no-deviation} uses the following simple but useful property of the rank function:
\begin{lem}
\label{lem:type-agreement}
Let $u \in \modelUniverse$ be an element with $\rank^i_u(\modelAlt,\modelAlt') = 0$ and let $\twoType \in \twoTypes(\signatureMain)$ with $\gaifmannRelation_i \in \twoType$.
For all $v \in \modelUniverse$ we have $\twoTypeOf^{\modelAlt|_\signatureMain}(u,v) = \twoType$ iff
$\twoTypeOf^{\modelAlt'|_\signatureMain}(u,v) = \twoType$.\\
(See Appendix~\ref{ap:lem:type-agreement} for the proof.)
\end{lem}

\begin{lem}\label{lem:sequence}
There is a sequence of $\signatureCTwoExtended$-structures $\modelAlt_0, \ldots, \modelAlt_p$,
with universe $\modelUniverse$ and $\modelAlt = \modelAlt_0$, such that:
\begin{enumerate}[(1)]
\item\label{enum:p1} $\oneTypeOf^{\modelAlt_i}(u) = \oneTypeOf^{\modelAlt'}(u)$ for all $u \in \modelUniverse$,
\item\label{enum:p2} $\twoTypeOf(\modelAlt_i) = \twoTypeOf(\modelAlt_{i+1})$ for all $0 \le i < p$,
\item\label{enum:p3} $\modelAlt_i \models \formulaMSOcopy$, (in particular $\modelAlt_i$ is $\messageAlphabet$-functional and chromatic),
\item\label{enum:p4} $\rank(\modelAlt_i,\modelAlt') > \rank(\modelAlt_{i+1},\modelAlt')$ for all $0 \le i < p$, and $\rank(\modelAlt_p,\modelAlt') = 0$.
\end{enumerate}
\end{lem}

\begin{proof}
Assume we have already defined $\modelAlt_i$ and $\rank(\modelAlt_i,\modelAlt') > 0$.
In the following we will define $\modelAlt_{i+1}$.
Because of $\rank(\modelAlt_i,\modelAlt') > 0$ we can choose some elements $u,v,w \in \modelUniverse$ and $j \in [k]$ such that $\modelAlt_i \models \gaifmannRelation_j(u,v)$, $\modelAlt' \models \gaifmannRelation_j(u,w)$ and $v \neq w$.
Let $\twoType = \twoTypeOf^{\modelAlt_i|_\signatureMain}(u,v)$.
We have $\twoType \in \MsgTypes{\messageAlphabet}(\signatureMain)$ because of $\gaifmannRelation_j \in \signatureOrientedKTree$.
We have $\modelAlt_i \models P_\twoType(u)$ because of $\modelAlt_i \models \definitionFormula_\unbounded$.
Because $\oneTypeOf(u)^{\modelAlt_i} = \oneTypeOf(u)^{\modelAlt'}$ we have $\modelAlt' \models P_\twoType(u)$.
Thus, $\modelAlt_i \models \typeFormula_{\twoType_y}(v)$ and $\modelAlt' \models \typeFormula_{\twoType_y}(w)$.
With $\oneTypeOf^{\modelAlt_i}(w) = \oneTypeOf^{\modelAlt'}(w)$ we get $\modelAlt_i \models \typeFormula_{\twoType_y}(w)$, and thus $\oneTypeOf^{\modelAlt_i|_\signatureMain}(v) = \oneTypeOf^{\modelAlt_i|_\signatureMain}(w)$.
We proceed by a case distinction:

\vspace{10pt}
\noindent
{\bf Case 1: $\twoType^{-1} $ is a $\messageAlphabet$-message type.}
\vspace{5pt}
\\
We have $\modelAlt' \models P_{\twoType^{-1}}(w)$ because of $\modelAlt' \models \definitionFormula_\bounded$.
We get $\modelAlt_i \models P_{\twoType^{-1}}(w)$ because of $\oneTypeOf^{\modelAlt_i}(w) = \oneTypeOf^{\modelAlt'}(w)$.
With $\modelAlt_i \models \definitionFormula_\unbounded$, there is an element $a \in \modelUniverse$ such that $\twoType = \twoTypeOf^{\modelAlt_i|_\signatureMain}(a,w)$ and $\modelAlt_i \models P_{\twoType}(a)$.
We note that $u \neq a$ because of $\modelAlt_i \models \gaifmannRelation_j(a,w)$, $v \neq w$ and $v$ is the unique element with $\modelAlt_i \models \gaifmannRelation_j(u,v)$ (using $\modelAlt_i \models \functionFormula$).
Moreover, $\modelAlt_i \models \typeFormula_{\twoType^{-1}_y}(a)$ and $\modelAlt' \models \typeFormula_{\twoType^{-1}_y}(u)$.
With $\oneTypeOf^{\modelAlt_i}(u) = \oneTypeOf^{\modelAlt'}(u)$ we get $\modelAlt_i \models \typeFormula_{\twoType^{-1}_y}(u)$, and thus $\oneTypeOf^{\modelAlt_i|_\signatureMain}(u) = \oneTypeOf^{\modelAlt_i|_\signatureMain}(a)$.


We note that the edges $(u,w)$, $(w,u)$, $(a,v)$ and $(v,a)$ do not have $\messageAlphabet$-message types (*) because $\modelAlt_i$ is chromatic, $u \neq a$, $v \neq w$, $\oneTypeOf^{\modelAlt_i|_\signatureMain}(v) = \oneTypeOf^{\modelAlt_i|_\signatureMain}(w)$, $\oneTypeOf^{\modelAlt_i|_\signatureMain}(u) = \oneTypeOf^{\modelAlt_i|_\signatureMain}(a)$ and the edges $(u,v)$, $(v,u)$, $(a,w)$ and $(w,a)$ have $\messageAlphabet$-message types.
Similarly, we get $\twoType \neq \twoTypeOf^{\modelAlt'|_\signatureMain}(u,v)$ and $\twoType \neq \twoTypeOf^{\modelAlt'|_\signatureMain}(a,w)$ (**) because $\modelAlt'$ is chromatic, $\twoType = \twoTypeOf^{\modelAlt'|_\signatureMain}(v,w)$    and $\twoType$ and $\twoType^{-1}$ are $\messageAlphabet$-message types.

We define $\modelAlt_{i+1}$ as follows:
The unary relations of $\modelAlt_{i+1}$ are defined such that we have $\oneTypeOf(u)^{\modelAlt_{i+1}} = \oneTypeOf^{\modelAlt_i}(u)$ for all elements $u \in \modelUniverse$.
We obtain the binary relations of $\modelAlt_{i+1}$ by swapping the edges $(u,w)$ and $(u,v)$ as well as $(a,w)$ and $(a,v)$ in $\modelAlt_i$:
we set
$\twoTypeOf^{\modelAlt_{i+1}|_\signatureMain}(u,w) = \twoType$,
$\twoTypeOf^{\modelAlt_{i+1}|_\signatureMain}(a,v) = \twoType$,
$\twoTypeOf^{\modelAlt_{i+1}|_\signatureMain}(a,w) = \twoTypeOf^{\modelAlt_i|_\signatureMain}(a,v)$ and
$\twoTypeOf^{\modelAlt_{i+1}|_\signatureMain}(u,v) = \twoTypeOf^{\modelAlt_i|_\signatureMain}(u,w)$;
these 2-types are well-defined because of $\oneTypeOf^{\modelAlt_i|_\signatureMain}(v) = \oneTypeOf^{\modelAlt_i|_\signatureMain}(w)$.
All other 2-types in $\modelAlt_{i+1}|_\signatureMain$ are the same as in $\modelAlt_i|_\signatureMain$.
This completes the definition of $\modelAlt_{i+1}$.

We now argue that $\modelAlt_{i+1}$ satisfies properties (\ref{enum:p1})-(\ref{enum:p4}).
Clearly, $\modelAlt_{i+1}$ satisfies (\ref{enum:p1}) by definition.
Because we only swapped edges from $\modelAlt_i$ to $\modelAlt_{i+1}$ we have (\ref{enum:p2}).
From (\ref{enum:p2}), Lemma~\ref{lem:star-type-equivalence} and $\modelAlt_i \models \chi$, we get $\modelAlt_{i+1} \models \chi$.
Because we only swapped swapped edges from $\modelAlt_i$ to $\modelAlt_{i+1}$
we get $\modelAlt_{i+1} \models \gaifmannFormula$ and $\modelAlt_{i+1} \models \containmentFormula$ from $\modelAlt_i \models \gaifmannFormula$ and $\modelAlt_i \models \containmentFormula$.
Because of (*) and $\twoTypeOf^{\modelAlt_i|_\signatureMain}(u,v) = \twoType = \twoTypeOf^{\modelAlt_i|_\signatureMain}(a,w)$ we get $\modelAlt_{i+1} \models \definitionFormula_\unbounded$ from $\modelAlt_i \models \definitionFormula_\unbounded$, that $\modelAlt_{i+1}$ is $\messageAlphabet$-functional because $\modelAlt_i$ is $\messageAlphabet$-functional and $\modelAlt_{i+1} \models \colorFormula$ from $\modelAlt_i \models \colorFormula$ (using that $\oneTypeOf^{\modelAlt_i|_\signatureMain}(v) = \oneTypeOf^{\modelAlt_i|_\signatureMain}(w)$
and
$\oneTypeOf^{\modelAlt_i|_\signatureMain}(u) = \oneTypeOf^{\modelAlt_i|_\signatureMain}(a)$).
Thus, $\modelAlt_{i+1}$ satisfies property (\ref{enum:p3}).

It remains to show property~(\ref{enum:p4}).
%
%
%
Using (*), (**) and Lemma~\ref{lem:type-agreement} we get that \linebreak
$\rank_u(\modelAlt_{i+1},\modelAlt') < \rank_u(\modelAlt_i,\modelAlt')$ and $\rank_z(\modelAlt_{i+1},\modelAlt') \le \rank_z(\modelAlt_i,\modelAlt')$ for $z \in \{v,w,a\}$.
Moreover, 
$\rank_z(\modelAlt_{i+1},\modelAlt') = \rank_z(\modelAlt_i,\modelAlt')$ for $z \in \modelUniverse \setminus \{u,v,w,a\}$.
Thus, we have $\rank(\modelAlt_i,\modelAlt') > \rank(\modelAlt_{i+1},\modelAlt')$.

\vspace{10pt}
\noindent
{\bf Case 2: $\twoType^{-1}$ is not a $\messageAlphabet$-message type.}
\vspace{5pt}
\\
We note that the edge $(w,u)$ does not have a $\messageAlphabet$-message type because $\modelAlt_i$ is chromatic, $v \neq w$, $\oneTypeOf^{\modelAlt_i|_\signatureMain}(v) = \oneTypeOf^{\modelAlt_i|_\signatureMain}(w)$ and the edge $(u,v)$ has a $\messageAlphabet$-message type.

We define $\modelAlt_{i+1}$ as follows:
The unary relations of $\modelAlt_{i+1}$ are defined such that we have
$\oneTypeOf(u)^{\modelAlt_{i+1}} = \oneTypeOf(u)^{\modelAlt_i}$ for all elements $u \in \modelUniverse$.
We obtain the binary relations of $\modelAlt_{i+1}$ by swapping edges in $\modelAlt_i$:
$\twoTypeOf^{\modelAlt_{i+1}|_\signatureMain}(u,w) = \twoType$, and
$\twoTypeOf^{\modelAlt_{i+1}|_\signatureMain}(u,v) = \twoTypeOf^{\modelAlt_i|_\signatureMain}(u,w)$.
These 2-types are well-defined because of  $\oneTypeOf^{\modelAlt_i|_\signatureMain}(v) = \oneTypeOf^{\modelAlt_i|_\signatureMain}(w)$.
All other 2-types in $\modelAlt_{i+1}|_\signatureMain$ are the same as in $\modelAlt_i|_\signatureMain$.
This completes the definition of $\modelAlt_{i+1}$.

We now argue that $\modelAlt_{i+1}$ satisfies properties (\ref{enum:p1})-(\ref{enum:p4}).
As in the previous case, one can argue that  $\modelAlt_{i+1}$ satisfies (\ref{enum:p1}) and (\ref{enum:p2}), $\modelAlt_{i+1} \models \chi$,          $\modelAlt_{i+1} \models \gaifmannFormula$ and $\modelAlt_{i+1} \models \containmentFormula$.
Because $(v,u)$ and $(w,u)$ do not have $\messageAlphabet$-message types we get $\modelAlt_{i+1} \models \definitionFormula_\unbounded$ from $\modelAlt_i \models \definitionFormula_\unbounded$ and that $\modelAlt_{i+1}$ is $\messageAlphabet$-functional because $\modelAlt_i$ is $\messageAlphabet$-functional.
We get $\modelAlt_{i+1} \models \colorFormula$ from $\modelAlt_i \models \colorFormula$, $\modelAlt' \models \colorFormula$ and $\oneTypeOf^{\modelAlt_i}(w) = \oneTypeOf^{\modelAlt'}(w)$.

It remains to show property~(\ref{enum:p4}).
From Lemma~\ref{lem:type-agreement} we get $\twoTypeOf^{\modelAlt_i|_\signatureMain}(u,w) \neq \twoType$ and $\twoTypeOf^{\modelAlt'|_\signatureMain}(u,v) \neq \twoType$.
Again applying Lemma~\ref{lem:type-agreement} we get $\rank_u(\modelAlt_{i+1},\modelAlt') < \rank_u(\modelAlt_i,\modelAlt')$.
Because $\twoTypeOf^{\modelAlt_i|_\signatureMain}(v,u)$ and $\twoTypeOf^{\modelAlt_i|_\signatureMain}(w,u)$ are not $\messageAlphabet$-message types, we get $\rank_v(\modelAlt_{i+1},\modelAlt') = \rank_v(\modelAlt_i,\modelAlt')$  and $\rank_w(\modelAlt_{i+1},\modelAlt') = \rank_w(\modelAlt_i,\modelAlt')$.
Moreover, we have that 
$\rank_z(\modelAlt_{i+1},\modelAlt') = \rank_z(\modelAlt_i,\modelAlt')$ for all $z \in \modelUniverse \setminus \{u,v,w\}$.
Thus, $\rank(\modelAlt_i,\modelAlt') > \rank(\modelAlt_{i+1},\modelAlt')$.

\end{proof}

\section{From Bounded Tree-width to Binary Trees}\label{se:translations}

This section is devoted to a discussion of the proof of Lemma~\ref{lem:MSO-tree-interp-idea}. 
First we need some background from the literature. 
A {\em\bf translation scheme for $\mcC_2$ over $\mcC_1$} is a tuple
$\t=\left\langle \phi,\psi_{C}:C\in\mcC_2\right\rangle$
of $\MS(\mcC_1)$-formulas such that $\phi$ has exactly one free first
order variable and the number of free first order variables in each $\psi_{C}$ is $\ar(C)$. 
The formulas $\phi$ and $\psi_{C}$, $C\in\mcC_2$, do
not have any free second order variables.\footnote{All translation schemes in this paper are scalar (i.e.\ non-vectorized).
In the notation of~\cite{bk:CourcelleEngelFriet12}, a translation
scheme is a parameterless non-copying $\MS$-definition scheme with
precondition formula $(x\approx x)$. }
The {\em\bf quantifier rank $\qr(\t)$  of $\t$} is the maximum of the quantifier ranks
of $\phi$ and the $\psi_{C}$. 
$\t$ is {\em\bf quantifier-free} if $\qr(\t)=0$. 
The {\em\bf induced transduction $\t^{\star}$} is a partial function
from $\mcC_1$-structures
to $\mcC_2$-structures which assigns a $\mcC_2$-structure $\t^\star(\mfA)$ to a $\mcC_1$-structure $\mfA$ as follows.
The universe of $\t^\star(\mfA)$ is 
$A_{\t}=\{a\in A:\,\mfA\models\phi(a)\}$. 
The interpretation of $C\in \mcC_2$ in $\t^\star(\mfA)$ is
$C^{\t^\star(\mfA)}=\left\{ \bar{a}\in{A_{\t}}^{\ar(C)}:\,\mfA\models\psi_{C}(\bar{a})\right\}$. 
Due to the convention that structures do not have an empty universe,
$\t^{\star}(\mfA)$ is defined iff $\mfA\models\exists x\phi(x)$. 

\begin{lem} [Fundamental property of translation schemes]\label{lem:fundamental}\label{lem:C2-trans}
Let $\t$ be a translation scheme for $\mcC_2$ over $\mcC_1$. 
There is a computable function 
$\t^{\sharp}$ from $\MS(\mcC_2)$-sentences to
$\MS(\mcC_1)$-sentences such that
for every $\mcC_1$-structure $\mfA$
for which $\t^{\star}(\mfA)$ is defined and for every $\MS(\mcC_2)$-sentence
$\theta$, 
$\mfA\models\t^{\sharp}(\theta)$ if and only if  $\t^{\star}(\mfA)\models\theta$.
$\t^\sharp$ is called the {\em\bf induced translation}. 
\end{lem}
For an $\MS(\mcC_2)$-sentence $\zeta$, $\t^{\sharp}$ 
substitutes the relation symbols
$C\in \mcC_2$ in $\zeta$ with the formulas $\psi_C$,
requires that each of the free variables satisfies $\phi$, 
and relativizes the quantification to $\phi$.
Appendix~\ref{ap:induced-trans} gives an inductive definition of $\t^\sharp$ following Definition 3.2 of~\cite{ar:MakowskyTARSKI}.

From the definition of the induced translation we have:
\begin{lem}[Quantifier-free translation schemes and $C^2$]\label{lem:fundamentalC2}
Let $\t$ be a quantifier-free translation scheme for $\mcC_2$ over $\mcC_1$. 
The induced translation $\t^{\sharp}$ maps $C^2(\mcC_2)$-formulas to $C^2(\mcC_1)$-formulas.
\end{lem}

It is well-known that the class of graphs of tree-width at most $k$ is the image of 
an induced transduction on a class of labeled trees~\cite{ar:ArnborgEtAl}. 
For the proof of Lemma~\ref{lem:MSO-tree-interp-idea} we need an analogous translation scheme
for $(\mcC_\bounded\cup\mcC_\unbounded)$-structures $\M$ whose reducts $\M|_{\mcC_\bounded}$ have tree-width at most $k$. 
In order that $\alpha$ and $\beta$ be mapped to an $\MS(\mcD_\bounded)$-sentence and a $C^2(\mcC_\unbounded)$-sentence respectively, 
we need that the translation scheme satisfy some additional properties. 
\begin{lem}\label{lem:translations-main}
Let $\mcC_\bounded$ and $\mcC_\unbounded$ be vocabularies
such that $\mcC_\bounded\cap \mcC_\unbounded$ only contains unary relation symbols. 
There  exist the following effectively computable objects: (1) a vocabulary $\mcD_\bounded$ consisting of the binary relation symbol $s$ and unary relation symbols only,
(2) a translation scheme $\tr = \left\langle \phi, \psi_C: C\in \mcC_\bounded \cup \mcC_\unbounded \right\rangle$ for $\mcC_\bounded\cup\mcC_\unbounded$ over $\mcD_\bounded\cup \mcC_\unbounded$, and
(3) an $\MS(\mcD_\bounded)$-sentence $\mathit{dom}$, such that:
\begin{enumerate}[(a)]
 \item\label{item:dom} $\phi$ is quantifier-free over $\mcD_\bounded$, 
 \item\label{item:C2} For every relation symbol $C\in \mcC_\unbounded$, $\psi_C$ is quantifier-free. 
 \item\label{item:MSO} For every relation symbol $C\in \mcC_\bounded$, $\psi_C$ is an $\MS(\mcD_\bounded)$-formula. 
 \item\label{item:K} Let $\mathscr{K}$ be the class of $(\mcD_\bounded\cup\mcC_\unbounded)$-structures in which $s$ is interpreted as a binary tree and 
 which satisfy $\mathit{dom}$. The image of $\mathscr{K}$ under $\tr^\star$ 
 is exactly the class of $(\mcC_\bounded\cup\mcC_\unbounded)$-structures $\M$ such that $\M|_{\mcC_\bounded}$ has tree-width at most $k$. 
\end{enumerate}
\end{lem}
The proof of Lemma~\ref{lem:translations-main} is technically similar to the discussion in~\cite{courcelle2002fusion}. 
We include the proof in the Appendix~\ref{ap:lem:MSO-tree-interp-idea} for completeness. 

We are  now  ready to prove Lemma~\ref{lem:MSO-tree-interp-idea}.
By Lemma~\ref{lem:translations-main}(\ref{item:K}) and Lemma~\ref{lem:fundamental}, statement~(i) in Lemma~\ref{lem:MSO-tree-interp-idea} holds
iff 
there is a $(\mcD_\bounded\cup \mcC_\unbounded)$-structure $\mfN$ such that $s^\mfN$ is a binary tree and 
$\mfN\models \mathit{dom}\land \tr^\sharp(\alpha)\land \tr^\sharp(\beta)$.
Let  $\alpha'=\mathit{dom}\land\tr^\sharp(\alpha)$ and $\beta'=\tr^\sharp(\beta)$. 
By Lemma~\ref{lem:translations-main}(\ref{item:MSO}) and the definition of $\tr^\sharp$, $\alpha' \in \MS(\mcD_\bounded)$.  
Let $\tr|_{\mcC_\unbounded}$ be the translation scheme for $\mcC_\unbounded$ over $\mcD_\bounded\cup \mcC_\unbounded$
which agrees with $\tr$ on all formulas, i.e. $\tr|_{\mcC_\unbounded} = \left\langle \phi, \psi_C: C\in \mcC_\unbounded\right\rangle$. 
The image of $\tr|_{\mcC_\unbounded}^\star$ on a structure $\M$ is the $\mcC_\unbounded$-reduct of the image of $\tr^\star$ on $\M$. 
Since $\beta \in C^2(\mcC_\unbounded)$, ${\tr|_{\mcC_\unbounded}}^\sharp(\beta)$ is well-defined 
and ${\tr|_{\mcC_\unbounded}}^\sharp(\beta) = \beta'$. By Lemma~\ref{lem:translations-main}(\ref{item:dom},\ref{item:C2}), 
$\tr|_{\mcC_\unbounded}$ is a quantifier-free
translation scheme, implying that $\beta' \in C^2(\mcD_\bounded\cup\mcC_\unbounded)$ by Lemma~\ref{lem:fundamentalC2}.

\section{From \texorpdfstring{$\MS$}{MSO} to \texorpdfstring{$C^{2}$}{C2} on Binary Trees}\label{se:hintikka}
The purpose of this section is to show that, on structures consisting only of a binary tree and additional unary relations, 
every $\MS$-sentence can be rewritten to a $C^2$-sentence 
which is equi-satisfiable and whose length is linear in the length of the input $\MS$-sentence. 
We start by introducing some tools for the literature.

\begin{theorem} [Hintikka sentences] \label{th:Hintikka-sentences} Let $\mcC$ be a vocabulary.
For every $q\in\mathbb{N}$ there is a finite set $\H_{\mcC,q}$ of
$\MS(\mcC)$ of quantifier rank $q$ such that: 
\begin{enumerate}
\item every $\epsilon\in\H_{\mcC,q}$ has a model; 
\item the conjunction of any two distinct sentences $\epsilon_{1},\epsilon_{2}\in\H_{\mcC,q}$
is not satisfiable; 
\item every $\MS(\mcC)$-sentence $\alpha$ of quantifier rank at most $q$
is equivalent to exactly one finite disjunction of sentences $\H_{\mcC,q}$; 
\item every $\mcC$-structure $\mfA$ satisfies exactly one
sentence $\hin_{\mcC,q}(\mfA)$ of $\H_{\mcC,q}$. We may omit $\mcC$ or $q$ from the subscript when they 
are clear from the context.
\end{enumerate}
\end{theorem}

For a class of $\mcC$-structures $\K$ an $n$-ary operation
$Op$ over $\mcC$-structures is called \emph{smooth}
\emph{over }$\K$, if for all $\mfA_{1},\ldots,\mfA_{n}\in\K$,
$\hin_{\mcC,q}(Op(\mfA_{1},\ldots,\mfA_{n}))$ depends only on
$\hin_{\mcC,q}(\mfA_{i})$: $i\in[n]$ and this dependence is computable\footnote{
Smooth operations here are called \emph{effectively smooth} in~\cite{ar:MakowskyTARSKI}.
}.
We omit \emph{``over
}$\K$\emph{'' }when $\K$ consists of
all $\mcC$-structures. 

\begin{theorem}[Smoothness]\label{th:smoothness}
\begin{enumerate}
\item The disjoint union is smooth. %
\item For every quantifier-free translation scheme $\t$, the
operation $\t^{\star}$ is smooth. 
\item Let $ \mfT_1 \Op \mfT_2$ denote the operation which augments the disjoint union of $\mfT_1$ and $\mfT_2$ by
adding an edge from the root of tree $\mfT_1$ to
the root of tree $\mfT_2$. This operation 
 is smooth over labeled trees.\footnote{We use the smoothness of the disjoint union and quantifier-free transductions. 
 The transduction adds the edge between the roots. 
Technically, in order for the transduction to be quantifier-free, 
we need that our trees have a special unary relation symbol 
$\rt$ with the natural interpretation as the root of the tree. }
%
\end{enumerate}
\end{theorem}

For an in-depth introduction to Hintikka sentences and smoothness 
and references to proofs see~\cite[Chapter 3, Theorem 3.3.2]{hodges1993model}
and~\cite{ar:MakowskyTARSKI} respectively. 

We are now ready for the main lemma of this section.

Let $\mathrm{rt}(x)$ be the sentence $\forall y\,\neg s(y,x)$. This sentence defines the root of the binary tree $s$. 
\begin{lem}\label{lem:HintikkaTree}
Let $q\in\mathbb{N}$ and let $\mcC$
be a vocabulary which consists only of the binary relation symbol $s$
and (possibly) additional unary relation symbols.
Let $C_{\epsilon}:\,\epsilon\in\H_{\mcC,q}$ be new unary
relation symbols. Let $\hin(\mcC,q)$ be the vocabulary which
extends $\mcC$ with $\{C_{\epsilon}:\,\epsilon\in\H_{\mcC,q}\}$.
There is a computable $C^{2}(\hin(\mcC,q))$-sentence $\psihin_{\mcC,q}^{\hin}$
such that if $\mfT_{0}$ is a $\mcC$-structure such that $s^{\mfT_0}$ is a binary tree:
\begin{enumerate}[(i)]
\item There is an expansion $\mfT_{1}$ of $\mfT_{0}$ to $\hin(\mcC,q)$ with
$\mfT_{1}\models\psihin_{\mcC,q}^{\hin}$. 
\item For every expansion $\mfT_{1}$ of $\mfT_{0}$ to $\hin(\mcC,q)$ 
and for every $\omega\in \MS(\mcC)$ 
such that \linebreak
$\mfT_{1}\models\psihin_{\mcC,q}^{\hin}$
and $\qr(\omega)=q$, 
the following holds:
there exists a $C^2$-sentence $\omega_{\hin}$ such that
%
$\mfT_{0}\models\omega$
%
iff 
$\mfT_{1}\models \omega_{\hin}$. 
The $C^2$-sentence $\omega_{\hin}$ is 
\[\forall x\left(\mathrm{rt}(x)\to \bigvee_{\epsilon\in \H_{\mcC,q}:\,\epsilon\models \omega}C_{\epsilon}(x)\right)
\]

%
\end{enumerate}
\end{lem}
The sentence $\psihin_{\mcC,q}^{\hin}$ is defined so that for every $\mfT_0$  there is a unique expansion $\mfT_1$
such that $\mfT_1 \models \psihin_{\mcC,q}^{\hin}$. For every $u$ in the universe $T_{1}$ of $\mfT_1$,
we will have  $u\in C_{\epsilon}^{\mfT_{1}}$ iff the subtree $\mfT_{u}$ of $\mfT_{0}$
whose root is $u$ satisfies $\mfT_{u}\models\epsilon$.
Using the smoothness of $\Op$, whether an element of $\mfT_1$ belongs to $C_\epsilon^{\mfT_1}$ depends only on its children.
This can be axiomatized in $C^2$. 
Lemma~\ref{lem:HintikkaTree-idea} follows from Lemma~\ref{lem:HintikkaTree} with
$q=\qr(\alpha)$,
$\mcD = \hin(\mcC,q)$, and 
$\omega_{\hin} = \gamma$.

Appendix~\ref{ap:lem:HT} spells out the proof of Lemma~\ref{lem:HintikkaTree}. 

\section{\texorpdfstring{$\MS$}{MSO} with Cardinality Constraints}\label{se:cardinality}

$\MS^{\card}$ denotes the extension of $\MS$ with atomic formulas
called \emph{cardinality constraints} $\sum_{i=1}^{r}|X_{i}|<\sum_{i=1}^{t}|Y_{i}|$,
where the $X_{i}$ and $Y_{i}$ are MSO variables,
and $|X|$ denotes the cardinality of $X$. 
Let $\WSone$ ($\WSone^{\card}$)
be the weak monadic second order theory (with cardinality constraints) of the
structure $\left\langle \mathbb{N},+1,<\right\rangle$. 
Let $\cardLogic \subseteq \MS^{\card}$  be the set of sentences $\rho$ such that 
(1) $\rho$ is of the form $\rho=\exists X_1 \cdots \exists X_m \omega$, and (2) only the $X_1,\ldots,X_m$ participate in cardinality
constraints. 
\begin{theorem}
\label{th:cardinality} Given a sentence $\rho \in \cardLogic$, 
it is decidable (A) whether  $ \left\langle \mathbb{N},+1,<\right\rangle \models \rho$, 
and (B) whether $\rho$ is satisfiable by a finite structure of
bounded tree-width. 
\end{theorem}
This theorem follows from Theorem~\ref{th:main}. 
The main observation needed for (B) is that cardinality
constraints can be expressed in terms of injective functions, which are axiomatizable in $C^{2}$. 
(A) is reducible to (B). The main observations for (A) are:
\begin{enumerate}
 \item that $X_1,\ldots,X_m$ are contained 
in the substructure $\mfA_1$ of $\left\langle \mathbb{N},+1,<\right\rangle$ generated by $\{0,\ldots,\ell\}$ for some $\ell\in \mathbb{N}$, 
\item that substructure $\mfA_2$ of $\left\langle \mathbb{N},+1,<\right\rangle$ generated by $\mathbb{N}-\{0,\ldots,\ell\}$
is isomorphic to $\left\langle \mathbb{N},+1,<\right\rangle$, and therefore $\mfA_2$ and $\left\langle \mathbb{N},+1,<\right\rangle$ have the same weak monadic 
second order theory,
\item that the weak monadic second order theory of $\left\langle \mathbb{N},+1,<\right\rangle$ is decidable, 
\item and that
$\left\langle \mathbb{N},+1,<\right\rangle$ is a transduction $\t$ of $\mfA_1\sqcup \mfA_2$.  
\end{enumerate}
It remains to use that $\sqcup$ is smooth, cf. Appendix~\ref{ap:th:cardinality}.

\newpage
\bibliographystyle{plain}

\newpage

 \section{Appendix}

\subsection{Proof of Theorem~\ref{th:cardinality}}\label{ap:th:cardinality}
We start with (B).

Let $\rho$ be a $\cardLogic$ sentence, i.e. 
the outermost block of quantifiers in $\rho$ is existential and only variables from the outermost
blocks may appear in cardinality constraints. 
For simplicity we consider $\rho=\exists X_{1}\exists X_{2}\omega$ with only two quantifiers in the outermost block. 
W.l.o.g.
the only cardinality constraint in $\omega$ is $|X_{1}|<|X_{2}|$.
By a slight abuse of notation, we sometimes treat $X_1$ and $X_2$ as unary relation symbols. 
Let $\mcC_{\bounded}$ extend the vocabulary of $\rho$ with new unary relation
symbols $X_{1}$,$X_{2}$,$W_{\img}$,$W_{\dom}$. Let $\mcC_{\unbounded}$
extend $\mcC_{\bounded}$ with a new binary relation symbol $B$. 
Finite satisfiability
of $\omega$ by a structure $\M$ such that $\tw(\M)\leq k$ can be
reduced to finite satisfiability of a sentence $\alpha\land\beta$, $\alpha\in\MS(\mcC_{\bounded})$
and $\beta\in C^{2}(\mcC_{\unbounded})$, by a structure $\mfA_{1}$ such that
$\tw(\mfA_{1}|_{\mcC_{\bounded}})\leq k$. Let $\beta$ be the $C^{2}$-sentence 
\[\beta=(\inj_{12}\lor\inj_{21})\land\dom\land\img
\]
where 
\begin{itemize}
 \item $\inj_{12}$ expresses that $B$ is an injective function from $X_{1}$
to $X_{2}$, 
\item $\inj_{21}$ espresses that $B$ is an injective function from
$X_{2}$ to $X_{1}$, 
\item $\dom$ expresses that the domain of $B$ is $W_{\dom}$,
\item and $\img$ expresses that the image of $B$ is $W_{\img}$. 
\end{itemize}

For every $\mcC_{\unbounded}$-structure
$\mfA_{1}$, $|X_{1}^{\mfA_{1}}|<|X_{2}^{\mfA_{1}}|$ iff $W_{\dom}^{\mfA_{1}}=X_{1}^{\mfA_{1}}$
and $X_{2}^{\mfA_{1}}\backslash W_{\img}^{\mfA_{1}}\not=\emptyset$. Let
$\alpha$ be obtained from $\omega$ by substituting every $|X_{1}|<|X_{2}|$
by 
\[\forall x\,(X_{1}(x)\leftrightarrow W_{\dom}(x))\land\exists x\:(\neg W_{\img}(x)\land X_{2}(x))
\]
For any $\mcC_{\bounded}$-structure $\mfA_{0}$ with $\tw(\mfA_{0})\leq k$,
$\mfA_{0}\models\omega$ iff there is an expansion $\mfA_{1}$ of $\mfA_{0}$
such that $\mfA_{1}\models\alpha\land\beta$. The treatment of other
cardinality constraints $\sum_{i=1}^{r}|X_{i}|<\sum_{i=1}^{t}|Y_{i}|$
is similar; it is helpful to assume w.l.o.g.\ that the $X_{i}$ and the $Y_{i}$
are disjoint.

Now we turn to (A).
The unary function $+1$ is the \emph{successor relation} of $\mathbb{N}$ and
interprets the binary relation symbol $\suc$. The binary relation
$<$ is the \emph{natural order relation} on $\mathbb{N}$.
In the proof of (A) we will use the theory of Hintikka sentences as presented in Section~\ref{se:hintikka}
with one caveat, namely that instead of restricting to finite structures, we allow arbitrary structures.
Theorems~\ref{th:Hintikka-sentences} and~\ref{th:smoothness} hold for arbitrary structures.
Theorems~\ref{th:Hintikka-sentences} guarantees the existence of a set 
$\H^\arb_{\mcC,q}$ analogously to $\H_{\mcC,q}$ for arbitrary structures. 
For every $\mcC$-structure $\mfA$, Theorem~\ref{th:Hintikka-sentences} guarantees
the existence of a sentence $\hin^\arb_{\mcC,q}(\mfA)$ of $\H^\arb_{\mcC,q}$
analogously to $\hin_{\mcC,q}(\mfA)$. For the rest of the section, we omit the superscript $\arb$
to simplify notation.

Consider $\rho=\exists X_{1}\exists X_{2}\omega$ for the vocabulary
$\mcC_{\mathbb{N}}$ of $\left\langle \mathbb{N},+1,<\right\rangle $.
Let $\alpha\in\MS(\mcC_{\bounded})$ and $\beta\in C^{2}(\mcC_{\unbounded})$ be as
discussed in the proof of (B) above. 
The following are equivalent:
\begin{enumerate}[1.]
\item $\left\langle \mathbb{N},+1,<\right\rangle \models\rho$
\item There are \emph{finite} unary relations $U_1$ and $U_2$
such that $\left\langle \mathbb{N},+1,<,U_1,U_2\right\rangle \models\omega$. 
$U_1$ and $U_2$ 
interpret $X_1$ and $X_2$ respectively. 
\item There are \emph{finite} unary relations $U_1$ and $U_2$
and an expansion $\mfA$ of $\left\langle \mathbb{N},+1,<,U_1,U_2\right\rangle $
such that $\mfA\models\alpha\land\beta$. $\mfA$ expands $\left\langle \mathbb{N},+1,<,U_1,U_2\right\rangle $
with interpretations of the symbols in $\mcD_\unbounded=\{B,W_{\dom},W_{\img}\}$.
\item $\rho'=\exists X_{1}\exists X_{2}(\alpha\land\beta)$ is satisfiable by an expansion
of $\left\langle \mathbb{N},+1,<\right\rangle$ with interpretations
for the symbols of $\mcD_\unbounded$.
\end{enumerate}

We have \emph{1.\ iff 2.} and \emph{3.\ iff 4.} by the semantics of $\exists X_{1}\exists X_{2}$
in weak monadic second order logic. We have \emph{2.\ iff 3.} similarly to the discussion
of $\alpha$ and $\beta$ in the proof of (A) above. The rest of the proof is devoted to proving that 4. is decidable. 

Observe that by the definition of $\beta$, and in weak monadic second order $X_1$ and $X_2$
are quantified to be finite sets, $B$ is axiomatized to
be a function with finite domain, and $W_\dom$ and $W_\img$ are finite. 
Hence models $\mfA$ of $\rho'$
can be decomposed into a finite part containing $B^{\mfA}$, and an
infinite part isomorphic to an expansion of $\left\langle \mathbb{N},+1,<\right\rangle $
in which the symbols of $\mcD_\unbounded$ as are interpreted as empty sets. We will use
a similar decomposition, but first we want to move from the structure
$\left\langle \mathbb{N},+1,<\right\rangle $ and its expansions to
$\left\langle \mathbb{N},+1\right\rangle $ and its expansions. There
is a translation scheme $\t_<$ such that for every structure
$\mfA = \left\langle \mathbb{N},+1, B^\mfA, W_\dom^\mfA, W_\img^\mfA \right\rangle$, 
$\t_<^{\star}(\left\langle \mfA \right\rangle)=\left\langle \mfA,<\right\rangle $,
where $\left\langle \mfA,<\right\rangle $ is the expansion of $\mfA$ with $<$. 
This is true since $<$ is $\MS$ definable from $+1$. 
We have that $\rho'$ is satisfied by an expansion of $\left\langle \mathbb{N},+1,<\right\rangle $
iff $\t_<^{\sharp}(\rho')$ is satisfied by
the same expansion of $\left\langle \mathbb{N},+1\right\rangle $.

Now we turn to the decomposition of models of $\t_<^{\sharp}(\rho')$
into a finite part containing $B^{\mfA}$, and an infinite part isomorphic
to a $\mcD_\unbounded$-expansion of $\left\langle \mathbb{N},+1\right\rangle$. 
Let $\mcD_{2}\supseteq\mcD_\unbounded$ be the vocabulary of $\t_<^{\sharp}(\rho')$.
For every $\mcD_{2}$-expansion $\P$ of $\left\langle \mathbb{N},+1\right\rangle $
and every $n\in\mathbb{N}$, let $\P_{1,n}$ and $\P_{n,\infty}$
be the substructures of $\P$ generated by $[n]$ and $\mathbb{N}\backslash[n-1]$
respectively. There is a translation scheme $\u$ such that $\u^{\star}(\P_{1,n}\sqcup\P_{n+1,\infty})=\P$
if $B^{\P}\subseteq[n]\times[n]$. $\u$ existentially quantifies
the set $[n]$ (which is the only non-empty finite set closed under
$\suc$ and its inverse), $1$ and $n$ (as the first and last elements
of $[n]$) and $n+1$ (as the only element without a $\suc$-predecessor
except for $1$) and adds the edge $(n,n+1)$ to $\suc$. We have
$\P\models\t_<^{\sharp}(\rho')$ iff $\P_{1,n}\sqcup\P_{n+1,\infty}\models\u^{\sharp}(\t_<^{\sharp}(\rho'))$.
Note that the vocabulary of $\u^{\sharp}(\t_<^{\sharp}(\rho'))$
is $\mcD_{2}$. Let $q$ be the quantifier rank of $\u^{\sharp}(\t_<^{\sharp}(\rho'))$.

The Hintikka sentence $\hin(\P_{n+1,\infty})$ of quantifier rank $q$ of $\P_{n+1,\infty}$ is uniquely defined since 
$\P_{n+1,\infty}$ is isomorphic to the expansion of $\left\langle \mathbb{N},+1\right\rangle $
with empty sets. Moreover, $\hin(\P_{n+1,\infty})$ is computable using that the theory 
of $\left\langle \mathbb{N},+1\right\rangle $ is decidable. 
Hence,
by the smoothness of the disjoint union, for every Hintikka sentence
$\epsilon\in\H_{\mcD_{2},q}$ there is a computable set $E_{\epsilon}\subseteq\H_{\mcD_{2},q}$
such that $\hin(\P_{1,n}\sqcup\P_{n+1,\infty})=\epsilon$ iff $\hin(P_{1,n})\in E_{\epsilon}$.
Then  $\P_{1,n}\sqcup\P_{n+1,\infty}\models\u^{\sharp}(\t_<^{\sharp}(\rho'))$ iff
$\P_{1,n}$ satisfies the sentence $\bigvee_{(\epsilon,\epsilon')}\epsilon'$,
where the $\bigvee_{(\epsilon,\epsilon')}$ ranges over pairs $(\epsilon,\epsilon')$
such that
\begin{enumerate}
\item $\epsilon\in\H_{\mcD_{2},q}$, 
\item $\epsilon\models\u^{\sharp}(\t_<^{\sharp}(\rho'))$, and 
\item $\epsilon'\in E_{\epsilon}$.
\end{enumerate}
Hence, 
$\rho'$ is satisfiable 
by an expansion
of $\left\langle \mathbb{N},+1,<\right\rangle$ 
iff $\bigvee_{(\epsilon,\epsilon')}\epsilon'$ is satisfiable by a finite structure in which $\suc$ is interpreted as a successor relation (i.e., as a simple directed path
on the whole universe). 
Let $\Pdef$ be the weak $\MS$-sentence such that the interpretation
of $\suc$ is a successor relation. By Theorem~\ref{th:cardinality}(B),
it is decidable whether, $\bigvee_{(\epsilon,\epsilon')}\epsilon'\land\Pdef$
is finitely satisfiable using that the class of simple directed paths annotated with unary relations has tree-width $1$. 

\begin{remark}
While we assumed for simplicity in (B) that $X_1$ and $X_2$ range over subsets of the universe, it is not hard to extend the proof to the case that 
$X_1$ and $X_2$ are guarded second order variables which range over subsets of any relation in the structure. This is true
since we can use the translation scheme $\tr$ from Lemma~\ref{lem:translations-main} to obtain the structures of tree-width at most $k$
as the image of $\tr^\star$ of labeled trees; $X_1$ and $X_2$ then  translate naturally to monadic second order variables. 
\end{remark}

\subsection{Proof of Lemma~\ref{lem:translations-main}}\label{ap:lem:MSO-tree-interp-idea}

This appendix is devoted to the proof of Lemma~\ref{lem:translations-main}. 
Appendix~\ref{ap:subse:aligned} introduces 
a tree encoding of structures of tree-width at most $k$ based on {\em aligned $k$-trees}. 
Appendix~\ref{ap:tr_mcC}
introduces a translation scheme such that the image of the induced transduction
on an $\MS$-definable class of labeled binary trees is 
 the class of structures of tree-width at most $k$. This is done based on
 translation schemes introduced in Appendix~\ref{ap:se:steps}. 
Finally, the proof of Lemma~\ref{lem:translations-main} is given in
 Subsection~\ref{ap:lem:translations-main:finally}. 

 Note that the notion of \emph{aligned $k$-trees} is quite similar to the notion of \emph{oriented $k$-trees} defined in 
 Section~\ref{se:background}. In fact, we could have replaced aligned $k$-trees with oriented $k$-trees here, or replaced 
 oriented $k$-trees with aligned $k$-trees throughout the paper. However, it simplifies the proofs
 to have two notions. 

\subsubsection*{Aligned \texorpdfstring{$k$}{k}-tree encodings and (decorated) aligned \texorpdfstring{$k$}{k}-trees}\label{ap:subse:aligned}

An aligned $k$-tree encoding $\mfT$
is a binary tree whose vertices are labeled by certain unary relations. 
All structures of bounded tree-width can be obtained by applying transductions 
to aligned $k$-tree encodings with additional vertex annotations in three steps:\vspace{3pt}

{\footnotesize
\begin{tabular}{llcl}
\ \ \ \bf Step A&\ \  \ \ \ \ \ \ \ \ \ Aligned $k$-tree encoding & $\stackrel{\transLONG}{\longrightarrow}$ & Decorated aligned $k$-tree 
\\
\ \ \ \bf Step B&\  \ \ \ \ \ \ \ \ \ \ Decorated aligned $k$-tree  & $\stackrel{\reduLONG}{\longrightarrow}$ & Aligned $k$-tree 
\\
\ \ \ \bf Step C&\  \ \ \ \ \ \ \ \ \ \ Aligned $k$-tree  & $\stackrel{\labLONG_\mcC}{\longrightarrow}$ & $\mcC$-structure of $\tw\leq k$\tabularnewline
\end{tabular}

}
\vspace{3pt}
$\transLONG$, $\reduLONG$ and $\labLONG_\mcC$ (shorthand $\trans$, $\redu$ and $\lab_\mcC$)
are unary operations on structures given in terms of translation schemes (see Section~\ref{ap:se:steps}).

Two key properties of our encoding:
\begin{enumerate}[(a)]
 \item Using that the number of edges in a $k$-tree is at most $k\cdot(\mbox{number of vertices})$,
our (decorated) aligned $k$-trees have functions $R_1,\ldots,R_k$ rather than an edge relation.
\item The elements of the universe of a structure appear as elements in the aligned $k$-tree, in the decorated aligned $k$-tree,  
and in the aligned $k$-tree encoding. Aligned $k$-tree encodings and decorated aligned $k$-trees have additional auxiliary vertices which are eliminated by
$\reduLONG$. 
\end{enumerate}
Let $\mcV$ be the vocabulary containing the binary relation symbol
$s$, and the unary relation symbols $\rt,\D_{1},\ldots,\D_{k},\D_{\blank}$. 

An {\em\bf aligned $k$-tree encoding} is a $\mcV$-structure $\mfT$ with universe
$T$ such that (i) $(T,s^{\mfT})$ is a directed tree (i.e., an acyclic directed graph
\begin{figure}
 {
\centering
\includegraphics[scale=0.6]{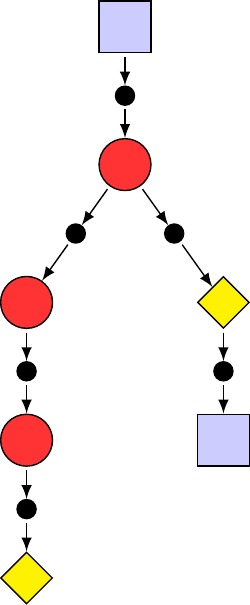}
\ \ \ \ \ \ \ 
\includegraphics[scale=0.6]{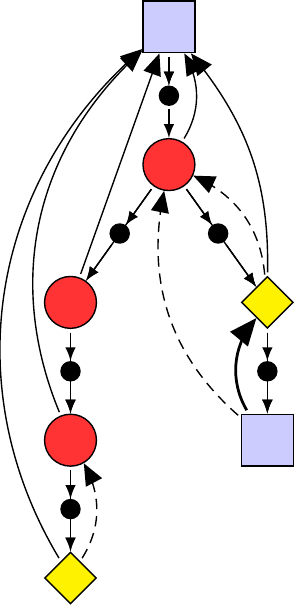}
\ \ \ \ \ \ \ 
\includegraphics[scale=0.6]{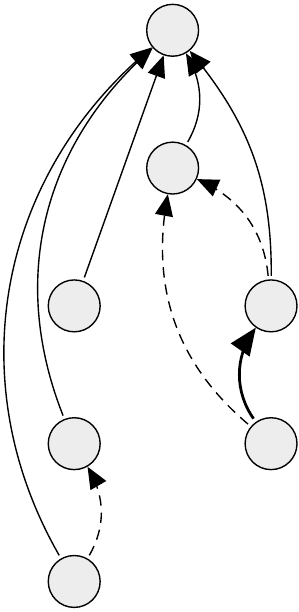}
\ \ \ \ \ \ \ 
\includegraphics[scale=0.6]{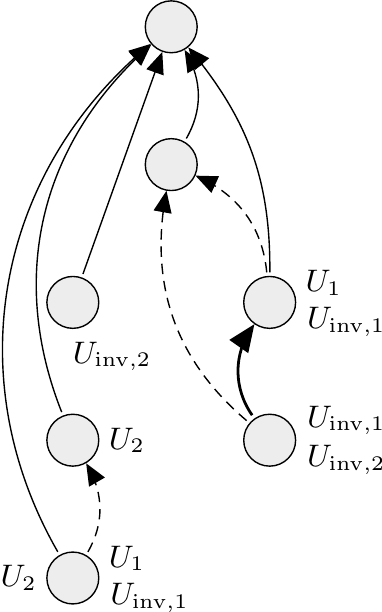}
\ \  \ 
\includegraphics[scale=0.6]{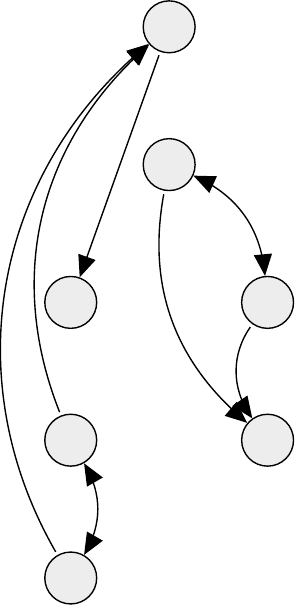}
\\\vspace{-2pt}
}
\ \ (a) \ \ \ \ \ \ \  \ \ \ \ \ \ \ \ \ \ \  
(b)\ \ \ \ \ \   \ \  \ \ \ \ \ \ \  \ \ \ \ \  
(c)  \ \  \ \  \ \ \ \ \ \ \  \ \ \ \    \  \ \ 
(d)\ \ \ \ \ \   \ \  \ \ \ \ \ \ \  \ \ \ \ \ \  
(e) \vspace{-5pt}
\caption{\label{fig-aligned} {\bf (a)} An aligned $3$-tree encoding. {\bf (b)} The decorated aligned $3$-tree obtained from (a) by applying $\trans^\star$. $R_1$, $R_2$ and $R_3$
are denoted by the dashed arrows, the solid arrows and the (single) thick arrow, respectively. 
{\bf (c)} The aligned $3$-tree obtained from (b) by applying $\redu^\star$. 
{\bf (d)} A aligned $3$-tree $\mfA$ expanding (c) with unary relations annotating its elements. $\mfA$ encodes a directed graph $G=(V,E)$.
The unary relations are $U_{j}^\mfA$ and $U_{\down,j}^\mfA$ for each $j\in \{1,2,3\}$. They encode,
for each element $m$ of $\mfA$, whether the edge $(m,R_j(m))$ respectively $(R_j(m),m)$ belongs to $E$. {\bf (e)}
The directed graph obtained from $\mfA$ in (d) by applying $\lab^\star$. $\gaif(G)$ has tree-width $2$. 
}
\end{figure}
in which all vertices have in-degree $1$ except exactly one.), 
(ii) every vertex in $(T,s^{\mfT})$ has out-degree at most~$2$, 
(iii) $\D_{1}^{\mfT},\ldots,\D_{k}^{\mfT},\D_{\blank}^{\mfT}$ form a partition of
$T$, (iv) $\rt^{\mfT}$ denotes the unary relation containing
only the unique vertex with in-degree $0$ in $(T,s^{\mfT})$, and 
(v) all the children of vertices in $T\backslash \D_\blank^\mfT$ 
belong to $\D_\blank^\mfT$.
Fig.~\ref{fig-aligned}(a) shows an aligned $3$-tree encoding.
The edges represent $s$. The small black circles represent $\D_{\blank}$, 
the larger red circles represent $\D_1$, the squares represent $\D_2$, and the diamonds represent $\D_3$.
\begin{lem}\label{lem:encoding-axiom}
There is a $C^2$-sentence $\enc$ such that for every $\mcV$-structure $\M$ in which $s^\M$ is a binary tree, 
$\model \models \enc$ iff $\model$ is an aligned $k$-tree encoding.
\end{lem}
\begin{proof}
Being an aligned $k$-tree encoding is expressible in $C^{2}$ for structures in which
$s$ is interpreted as a binary tree. Below $\xi_{X}$ refers to
requirement X in the definition of an aligned $k$-tree encoding.
$\xi_{i}$ and $\xi_{ii}$ are already taken care of since $s$
is a binary tree.
$\enc$ is given as follows:
\begin{eqnarray*}
\enc & = & \xi_{iii}\land\xi_{iv}\land\xi_{v}\\
\xi_{iii} & = & \forall x\left(\D_{\blank}(x)\lor\bigvee_{i=1}^{k}\D_{i}(x)\right)\land\\
 &  & \forall x\left(\bigvee_{i=1}^{k}\D_{i}(x)\leftrightarrow\neg \D_{\blank}(x)\right)\land\\
 &  & \forall x\left(\bigwedge_{i\not=j}\left(\neg \D_{i}(x)\lor\neg \D_{j}(x)\right)\right)\\
\mathrm{\xi_{iv}} & = & \forall x\left(\rt(x)\leftrightarrow\forall y\neg s(y,x)\right)\\
\xi_{v} & = & \forall x\forall y\left(s(x,y)\land\neg \D_{\blank}(x)\to \D_{\blank}(y)\right)
\end{eqnarray*}
\end{proof}

A {\em\bf decorated aligned $k$-tree} $\mfA$ is an expansion of an aligned $k$-tree encoding $\mfT$
to $\signatureOrientedKTree \cup \mcR$. For every $j\in[k]$, $R_{j}^{\mfA}$ contains all pairs
$(v,u)\in(T\backslash (\D_{j}^\mfT\cup \D_{\blank}^\mfT))\times \D_{j}^{\mfT}$
such that there is a directed path $P$ from $u$ to $v$ in $(T,s^{\mfT})$
which does not intersect with $\D_{j}^{\mfT}$ except on $u$, i.e.\ $\D_{j}^{\mfT}\cap(P\backslash\{u\})=\emptyset$;
moreover, for every $u$, if $R_j^\mfA$ does not contain any pair $(u,v)$, then $R_j^{\mfA}$ contains additionally
the self loop $(u,u)$. 
Observe that the relation $R_{j}^{\mfA}$ is a total function.
Fig.~\ref{fig-aligned}(b) shows the decorated aligned $3$-tree obtained from (a). 
The new edges represent $R_1$, $R_2$, and $R_3$: dashed edges represent $R_1$, 
solid edges represent $R_2$, and the single thick edge (from a square to a diamond) represents $R_3$. 

A {\em\bf aligned $k$-tree\footnote{
A small but important note is that the definition of aligned $k$-trees
used throughout this appendix deviates slightly from
the one given in Section~\ref{se:background} and used elsewhere in the paper. 

}}
is the $\mcR$-reduct of the substructure 
of a decorated aligned $k$-tree $\mfA_0$
generated by $\D_1^{\mfA_0}\cup\cdots\cup \D_k^{\mfA_0}$. 
Fig.~\ref{fig-aligned}(c) shows the aligned $3$-tree obtained from (b).
The vertices of (c) are not labeled by any unary relation.

\begin{lem}\label{lem:o-k-tree}~ 
\begin{enumerate}[(1)]
\item For every aligned $k$-tree $\mfA$, $\gaif(\mfA)$ is a partial
$(k-1)$-tree. 
\item For every partial $(k-1)$-tree $G$ there is an aligned $k$-tree
$\mfA$ such that $G$ is a subgraph of $\gaif(\mfA)$. 

\end{enumerate}
\end{lem} 

\begin{proof}
For the proof of Lemma~\ref{lem:o-k-tree}(1) we define
an aligned $k$-tree encoding $\mfT$ to be {\em\bf perfect} if (1) there
are $w_{1},p_1,w_2,\ldots,p_{k-1},w_{k}\in T$ such that $w_{1},p_1,w_2,\ldots,p_{k-1},w_{k}$ is a directed path, (2) $\rt^{\mfT}=\{w_{1}\}$,
 (3) $w_{i}\in \D_{i}^{\mfT_{0}}$ of all $i\in[k]$, (4) $p_i\in \D_{\blank}$ for all $i\in[k-1]$, and (5) the out-degree
of each of vertex in the path is exactly $1$.
If $\mfT$ is a perfect aligned $k$-tree encoding, then
$\redu^\star(\trans^\star(\mfT))$ is a
perfect  aligned $k$-tree. We have:

\vspace{5pt}
($\diamond$) For every perfect  aligned $k$-tree $\mfA$, $\gaif(\mfA)$ is
a $(k-1)$-tree.
\vspace{5pt}

($\diamond$) is proved by induction on the number of vertices in a
perfect aligned $k$-tree encoding $\mfT$ such that $\redu^\star(\trans^\star(\mfT))=\mfA$.
Lemma~\ref{lem:o-k-tree}(1) follows from ($\diamond$) by building
a perfect  aligned $k$-tree such that $\mfA$ is its substructure.

First we prove:
\begin{quote}
($\diamond$) For every perfect  aligned $k$-tree $\mfA$, $\gaif(\mfA)$ is
a $(k-1)$-tree.
\end{quote}

Let $\mfT$ be a perfect aligned $k$-tree encoding such that $\redu^{\star}(\trans_{}^{\star}(\mfT))=\mfA$.
We prove the claim by induction on the number of vertices in $\mfT$.

\begin{itemize}
\item In the base case, $\mfT$ consists only of the vertices $w_{1},p_1,\ldots,p_{k-1},w_{k}$,
and $\redu^{\star}(\trans_{}^{\star}(\mfT))$ is a $k$-clique, which
is a $(k-1)$-tree.
\item Let $\mfT$ be a perfect aligned $(k-1)$-tree encoding. Let $\mfT_{0}$
be obtained from $\mfT$ by removing a leaf $v$. By the induction hypothesis,
$\gaif(\trans_{}^{\star}(\mfT_{0}))$ is a $(k-1)$-tree. If $v\in \D_{\blank}^{\mfT}$,
then $\redu^{\star}(\trans_{}^{\star}(\mfT_{0}))=\redu^{\star}(\trans_{}^{\star}(\mfT))$.
Otherwise, let $v\in \D_{j}^{\mfT}$. We denote by $B=\{u_{i}:i\in[k]\backslash\{j\}\}$
the set of elements such that $u_{i}\in \D_{i}$ and there is a path
$P_{i}$ from $u_{i}$ to $v$ in $\mfT$ satisfying the following property:
$P_{i}$ does not visit any vertex in $\D_{i}$ except $u_{i}$. Note
that for all $i\in[k]\backslash\{j\}$ $u_{i}$ exists, and hence
$|B|=k-1$, using that $\mfT_{0}$ is perfect. By the choice of $B$
and since $v$ is a leaf, $B$ is the set of neighbors of $v$ in
$\gaif(\redu^{\star}(\trans^{\star}(\mfT)))$. It remains to show that
$B$ is a $(k-1)$-clique to get that $\gaif(\redu^{\star}(\trans_{}^{\star}(\mfT)))$
is a $(k-1)$-tree. Let $i_{1},i_{2}\in[k]\backslash\{j\}$ with $i_{1}\not=i_{2}$.
Since $(T,s^{\mfT})$ is a directed tree, either $P_{i_{1}}\subsetneq P_{i_{2}}$
or $P_{i_{2}}\subsetneq P_{i_{1}}$. W.l.o.g.\ let $P_{i_{1}}\subsetneq P_{i_{2}}$,
so $P_{i_{2}}$ is a path from $u_{i_{2}}$ to $v$ which visits $u_{i_{1}}$.
Let $P_{21}\subsetneq P_{i_{2}}$ be the path from $u_{i_{2}}$ to
$u_{i_{1}}$. By the choice of $B$, there is no vertex in $P_{21}$
which belongs to $\D_{i_{2}}$ except $u_{i_{2}}$. Hence, $(u_{i_{1}},u_{i_{2}})\in R_{i_{2}}^{\redu^{\star}(\trans_{}^{\star}(\mfT))}$,
i.e.\ there is an edge between $u_{i_{1}}$ and $u_{i_{2}}$ in $\gaif(\redu^{\star}(\trans_{}^{\star}(\mfT)))$.
We get that $B$ is a $(k-1)$-clique and $\gaif(\redu^{\star}(\trans_{}^{\star}(\mfT)))$
is a $(k-1)$-tree.
\end{itemize}

Now we are ready to prove Lemma~\ref{lem:o-k-tree}(1) and Lemma~\ref{lem:o-k-tree}(2).

\begin{enumerate}[(1)]
\item Let $\mfT$ be an aligned $k$-tree encoding such that $\redu^{\star}(\trans_{}^{\star}(\mfT))=\mfA$.
Let $\mfT_{0}$ be obtained from the directed path $w_{1},p_1,\ldots,p_{k-1},w_{k}$
with $w_{i}\in \D_{i}^{\mfT_{0}}$ for each $i\in[k]$,
$p_i\in \D_{\blank}^{\mfT_0}$ for each $i\in[k-1]$,
 and $\rt^{\mfT}=\{w_{1}\}$
by attaching $\mfT$ as a child of $w_{k}$. By construction, $\mfT_{0}$
is perfect. By ($\diamond$) we have that $\gaif(\redu^{\star}(\trans_{}^{\star}(\mfT_{0})))$
is a $(k-1)$-tree. Deleting the vertices $w_{1},p_1,\ldots,p_{k-1},w_{k}$ and
the edges incident to them from $\gaif(\redu^{\star}(\trans_{}^{\star}(\mfT_{0})))$
we get $\gaif(\redu^{\star}(\trans^{\star}(\mfT)))$, implying that
$\gaif(\redu^{\star}(\trans_{}^{\star}(\mfT)))$ is a partial $k$-tree.
\item We prove this by induction on the construction of $(k-1)$-trees.

\begin{enumerate}
\item Let $G$ be a $(k-1)$-clique with vertices $w_{1},\ldots,w_{k-1}$.
Let $\mfT_{base}$ be the directed path on the vertices $w_{1},p_1,\ldots,p_{k-2},w_{k-1}$
such that $w_{i}\in \D_{i}^{\mfT_{base}}$ for all $i\in[k-1]$
and $p_i\in \D_{\blank}^{T_{base}}$ for all $i\in[k-2]$. We have
that $\gaif(\redu^{\star}(\trans_{}^{\star}(\mfT_{base})))=G$.
\item Let $G_{0}$ be a $(k-1)$-tree and let $\mfT_{0}$ be 
an aligned $k$-tree encoding such that $\gaif(\redu^{\star}(\trans_{}^{\star}(\mfT_{0})))=G_{0}$.
Let $C=\{u_{1},\ldots,u_{k-1}\}$ be a $(k-1)$-clique in $G_{0}$.
Let $G$ be the graph obtained from $G_{0}$ by adding a new vertex
$v$ as well as edges between $v$ and each of the vertices in $C$.
Since $C$ is a $(k-1)$-clique in $G_{0}$, for every distinct $i_{1},i_{2}\in[k-1]$
there is a directed path in $\mfT_{0}$ between from $u_{i_{1}}$ to
$u_{i_{2}}$ or vice versa. As a consequence, there is a path $P$
in $\mfT_{0}$ from the root to some $u_{\ell}$ such that $u_{1},\ldots,u_{k-1}$
all occur on $P$. For every $t\in[k-1]$ and $i\in[k]$ such that
$u_{t}\in \D_{i}^{\mfT_{0}}$, $u_{t}$ is the last vertex in $P$ to
belong to $\D_{i}^{\mfT_{0}}$. Since $|C|=k-1$, there is $j\in[k]$
such that $\D_{j}^{\mfT_{0}}\cap C=\emptyset$. Since $C$ is a clique
in $\gaif(\redu^{\star}(\trans_{}^{\star}(\mfT_{0})))$, it must be
the case that $|\D_{i}^{\mfT_{0}}\cap C|=1$ for all $i\in[k]\backslash\{j\}$.
If $u_{\ell}$ has two children, let $\mfT$ be a subtree of $\mfT_{0}$
whose root is a child of $u_{\ell}$. Let $\mfT$ be obtained from $\mfT_{0}$
by attaching a new child vertex $v_{\blank}\in \D_{\blank}^{\mfT}$ to
$u_{\ell}$, moving $\mfT$ to be a child of $v_{\blank}$, and adding
$v$ as a child of $v_{\blank}$. If $u_{\ell}$ is a leaf or has
one child in $\mfT_{0}$, let $\mfT$ be obtained from $\mfT_{0}$ by attaching
$v$ as a child of $u_{\ell}$. Setting $v$ to belong to $\D_{j}^{\mfT}$,
we get $\gaif(\redu^{\star}(\trans^{\star}(\mfT)))=G$.
\end{enumerate}
\end{enumerate}

\end{proof}

%
%

\subsection*{The translation schemes \texorpdfstring{$\transLONG$}{interpret}, \texorpdfstring{$\reduLONG$}{undecorate} and \texorpdfstring{$\labLONG_\mcC$}{structurize\_C}}\label{ap:se:steps}

There are translation schemes $\transLONG$
(shorthand $\trans$) for $\signatureOrientedKTree\cup \mcV$ over $\mcV$
and $\reduLONG$ (shorthand $\redu$) for $\mcR$ over $\signatureOrientedKTree\cup \mcV$
which take an aligned $k$-tree encoding to its decorated aligned $k$-tree
respectively a decorated aligned $k$-tree to its aligned $k$-tree. 
The induced transductions $\trans^\star$ and $\redu^\star$
are surjective with respect to the classes of decorated aligned $k$-trees resp.\ aligned $k$-trees.

For every vocabulary $\mcC$ there is
a translation scheme $\labLONG_\mcC$ (shorthand $\lab_\mcC$) which takes 
an aligned $k$-tree whose elements are annotated with unary relations
to a $\mcC$-structure. The induced transduction $\lab_\mcC$ is surjective with respect to
the class of $\mcC$-structures of tree-width less than $k$. The binary relations of such structures are 
encoded inside aligned $k$-trees by
unary relations on the sources of $R_j$ edges, using that the $R_j$ are functions. 
The vocabulary consisting of these new unary relation symbols as well as the unary relation symbols
of $\mcC$ 
is denoted by $\mcN_\mcC$.
$\lab_\mcC$ is a translation scheme for $\mcC$ over $\mcR \cup \mcN_\mcC$. 
Fig.~\ref{fig-aligned}(d) shows an aligned $k$-tree annotated with the unary relations of $\mcN_\mcC$ for $\mcC= \left\langle E \right\rangle$,
where $E$ is a binary relation symbol. Fig.~\ref{fig-aligned}(e) shows the directed graph obtained from (d) by applying $\lab_\mcC$. 

\begin{lem}\label{lem:translations}\label{lem:trans--aaaa}\label{lem:redu--aaaa}\label{lem:lab---aaaa}\ \\ \vspace{-15pt}
\begin{description}
 \item[\bf \ \ \ A.] For every decorated aligned $k$-tree $\mfA$, $\trans^\star(\mfA|_{\mcV}) = \mfA$.
\item[\bf \ \ \ B.]  For every aligned $k$-tree $\mfA$ obtained from a decorated aligned $k$-tree $\mfA_0$,
$\redu^\star(\mfA_0)=\mfA$. 
\item[\bf \ \ \ C.]  
There is a vocabulary $\mcN_\mcC$ consisting of unary relation symbols only, such that for
every $\mcC$-structure $\M$, $\M$ has tree-width at most $k$ iff 
there is an aligned $k$-tree $\mfA_{\mathit{ori}}$ and an expansion $\mfA$ of $\mfA_{\mathit{ori}}$ to $\Xi_\mcC \cup\mcR$
such that $\lab_\mcC^{\star}(\mfA)=\M$.
\end{description}
\end{lem}

Lemma~\ref{lem:MSO-tree-interp-idea} follows from Lemma~\ref{lem:translations}, Lemma~\ref{lem:encoding-axiom}, and Lemma~\ref{lem:fundamental}
with $\mcD_\bounded = \mcN_{\mcC_\bounded} \cup \mcV$, 
$\alpha' = \enc \land \trans^{\sharp}(\redu^\sharp(\lab_\mcC^\sharp(\alpha)))$, and
$\beta' = \trans^{\sharp}(\redu^\sharp(\lab_\mcC^\sharp(\beta)))$. 

\subsection*{Proof of Lemma~\ref{lem:translations}(A)}\label{ap:lem:trans--}
$\trans=\left\langle \varphi,\psi_{C}:C\in\signatureOrientedKTree\right\rangle$ is given as follows:
\begin{enumerate}
\item The universe $A$ of $\mfA$ is $T$, i.e.\ $\varphi(x)= (x\approx x) $.
\item The relations from $\mcV$ are copied without change, i.e.\ $\psi_{s}(x,y)=s(x,y)$,
and $\psi_{\D}(x)=\D(x)$ for each $\D\in\{\D_{1},\ldots,\D_{k},\D_{\blank}\}$.
\item
$
\psi_{R_{j}}(x,y) = \varphi(x)\land\neg \D_{j}(x)\land \D_{j}(y)\land\forall z\left(z\not\approx y\to\neg \D_{j}(z)\right)\land\left(\dpath(x,y)\right)
$
where $\dpath(x,y)$ says that there is a subset $P$ of $s$ which
is a directed simple path from $x$ to $y$:
\begin{eqnarray*}
\dpath(x,y) & = & \exists P\subseteq s\,\Big[\forall z\left(in_{\leq1}(P,x)\right)\land in_{0}(P,x)\land\\
 &  & \hphantom{\exists P\subseteq s\,\Big[}\forall z\left(in_{\leq1}(P,y)\right)\land out_{0}(P,y)\land\reach(P,x,y)\Big]\\
\reach(P,x,y) & = & \forall X\subseteq V\Bigg(\left(X(x)\land\neg X(y)\right)\to\\
&&\hphantom{\forall X\subseteq V\Bigg(}\left(\exists z_{1}\exists z_{2}\left(X(z_{1})\land\neg X(z_{2})\land P(z_{1},z_{2})\right)\right)\Bigg)\\
in_{0}(z_{1},P) & = & \forall z_{2}\left(\neg P(z_{2},z_{1})\right)\\
out_{0}(z_{1},P) & = & \forall z_{2}\left(\neg P(z_{1},z_{2})\right)\\
in_{\leq1}(z_{1},P) & = & \exists^{\leq1}z_{2}\left(P(z_{2},z_{1})\right)\\
out_{\leq1}(z_{1},P) & = & \exists^{\leq1}z_{2}\left(P(z_{1},z_{2})\right)
\end{eqnarray*}
$\reach(P,x,y)$ says that $y$ is reachable from $x$ via $P$ edges.
\end{enumerate}

\subsection*{Proof of Lemma~\ref{lem:translations}(B)}\label{ap:lem:redu--}
Let $\redu=\left\langle \varphi,\psi_{R_{1}},\ldots,\psi_{R_{k}}\right\rangle$
be as follows: $\varphi(x)=\neg \D_{\blank}(x)$ and for
every $i\in[k]$, $\psi_{R_{i}}(x,y)=R_{i}(x,y)$.

\subsection*{Proof of Lemma~\ref{lem:translations}(C)}\label{ap:lem:lab}
We prove the following lemma which spells out the two directions of Lemma~\ref{lem:translations}(C):
\begin{lem}\label{lem:lab2}
For every vocabulary $\mcC$, there
is a set $\Xi_{\mcC}$ of unary relation symbols and
a translation scheme $\labLONG_\mcC$ (shorthand $\lab_\mcC$) for $\mcC$ over $\Xi_{\mcC}\cup\mcR$
such that:
\begin{enumerate}[(1)]
\item If $\M$ is a $\mcC$ structure whose Gaifman graph is a partial $(k-1)$-tree,
then there is an expansion $\mfA$ of an aligned $k$-tree to
$\Xi_{\mcC}\cup\mcR$ for which $\lab_\mcC^{\star}(\mfA)=\M$.
\item If $\mfA$ is a $(\Xi_\mcC \cup\mcR)$-structure and $\mfA_{\mathit{red}}=\mfA|_\mcR$ is an aligned $k$-tree,
then $\gaif(\lab_\mcC^{\star}(\mfA))$ is a partial $(k-1)$-tree.
\end{enumerate}
\end{lem}

For every vocabulary $\mcC$, let $\Xi_{\mcC}$
be the vocabulary extending the set $\unary(\mcC)$ of unary relation symbols in $\mcC$ by
fresh unary relation symbols
$U_{B,\self}$, $U_{B,j}$, and $U_{B,\down,j}$ for every binary relation symbol $B\in\mcC$
and $j\in[k]$.
Recall $\mcV_\mcC = \mcV\cup\Xi_\mcC$ and $\signatureOrientedKTree_{\mcC}=\mcV_{\mcC}\cup\mcR\cup\mcC$.

Since the $R_i$ are functions, we use a fixed number of unary relation symbols to encode,
for every universe element $m_1$ such that $(m_1,m_2)\in R_j^\M$,
what other relations of $\M$  do $(m_1,m_2)$ and $(m_2,m_1)$ belong to.
The unary relation symbols are of the form $U_{B,j}$ and $U_{B,\down,j}$, where $j\in[k]$
and $B$ is a binary relation symbol.
We use that tuples $(m_1,m_2)$ which occur in any relation of $\M$, occur
also in $R_j^\M$ or in ${(R_j^\M)}^{-1}$. Additional unary relation symbols $U_{B,\self}$ encode self loops $(m,m)$
in $B^\M$.

Let $\lab_\mcC=\left\langle \varphi,\psi_{C}:C\in\mcC\right\rangle $ is
the translation scheme given as follows.
\begin{itemize}
\item $\varphi(x)= (x\approx x)$.
\item For every unary relation symbol $U\in \mcC$, $\psi_U(x)=U(x)$.
\item For every binary relation symbol $B\in\mcC$, $B^{\lab_\mcC^{\star}(\mfT)}$
consists of all pairs $(v,u)$ such that (a) $(v,u)\in R_{j}^{\mfT}$
and $v\in U_{B,j}^{\mfT}$ or (b) $(u,v)\in R_{j}^{\mfT}$ and $u\in U_{B,\down,j}^{\mfT}$.
This is given by
\begin{eqnarray*}
\psi_{B}(x,y)&=&\neg(x\approx y)\land \bigvee_{j\in[k]}\left(\left(R_{j}(x,y)\land U_{B,j}(x)\right)\lor\left(R_{j}(y,x)\land U_{B,\down,j}(y)\right)\right)
\\
&&\lor(x\approx y)\land U_{B,\self}(x)
\end{eqnarray*}

Let $\M$ be a $\mcC$-structure with universe $M$. Let $\mfA_{0}$
be an aligned $k$-tree such that $\gaif(\M)=\gaif(\mfA_{0})$
guaranteed by Lemma~\ref{lem:o-k-tree}. Let $\mfA_{1}$ be the expansion
of $\mfA_{0}$ such that, for every $B\in\mcC$ and $j\in[k]$,
\begin{eqnarray*}
U_{B,j}^{\mfA_{1}} & = & \left\{ m_{1}\mid(m_{1},m_{2})\in R_{j}^{\mfA_{1}}\cap B^{\M}\right\} \\
U_{B,\down,j}^{\mfA_{1}} & = & \left\{ m_{1}\mid(m_{2},m_{1})\in {(R_{j}^{A_{1}})}^{-1}\cap B^{\M}\right\} \\
U_{B,\self}^{\mfA_1} &=& \{m\mid (m,m)\in B^{\M}\}
\end{eqnarray*}
Self loops are encoded by the relations of the form  $U_{B,\self}^{\mfA_1}$.
Consider $a,b\in M$ which are distinct. We divide into cases. 
\begin{itemize}
\item Assume $(a,b)$ is not an edge of $\gaif(\M)$. By $\gaif(\M)=\gaif(\mfA_{1})$,
neither $(a,b)$ nor $(b,a)$ belong to any $R_{j}^{\mfT_{0}}=R_{j}^{\mfT_{1}}$.
Hence, $\mfA_{1}\not\models\psi_{B}(a,b)$, i.e.\ $(a,b)\notin B^{\lab_\mcC^{\star}(\mfA_{1})}$,
for any binary relation symbol $B\in\mcC$.
\item Assume $(a,b)$ is an edge of $\gaif(\M)$. By $\gaif(\M)=\gaif(\mfA_{0})$,
either $(a,b)$ or $(b,a)$ belong to some $R_{j}^{\mfA_{0}}=R_{j}^{\mfA_{1}}$.
We have: $(a,b)\in B^{\M}$ iff exists $j\in[k]$ such that [$(a,b)\in R_{j}^{\mfA_{1}}$
and $a\in U_{B,j}^{\mfA_{1}}$] or [$(b,a)\in R_{j}^{\mfA_{1}}$
and $b\in U_{B,\down,j}^{\mfA_{1}}$] iff $\mfA_{1}\models\psi_{B}(a,b)$
iff $(a,b)\in B^{\lab_\mcC^{\star}(\mfA_{1})}$.
\end{itemize}

In both cases $(a,b)\in B^{\M}$ iff $(a,b)\in B^{\lab_\mcC^{\star}(\mfA_{1})}$.
We get that $\M=\lab_\mcC^{\star}(\mfA_{1})$ and (1) follows.

\end{itemize}

Now we turn to (2). By Lemma~\ref{lem:o-k-tree}(1), the Gaifman graph
of $\mfA$ is a partial $(k-1)$-tree. By definition of $\psi_{B}$,
for every binary $B\in\mcC$, $B^{\lab_\mcC^{\star}(\mfA)}\subseteq\bigcup_{j\in[k]}R_{j}^{\mfA}\cup{(R_{j}^{\mfA})}^{-1}$,
implying that $\gaif(\lab_\mcC^{\star}(\mfA))$ is a partial $(k-1)$-tree.

\subsection*{The translation scheme \texorpdfstring{$\tr_\mcC$: \\}{tr\_C:}  from structures of bounded tree-width to labeled trees}\label{ap:tr_mcC}

Let $\trans_\mcC$ and $\redu_\mcC$ be the translation schemes which extend $\trans$ and $\redu$ 
respectively as follows. The unary relations $C$ of $\mcN_\mcC$ are additionally 
defined under $\trans_\mcC$ and $\redu_\mcC$ to be $\psi_C(x)=C(x)$. Hence, 
$\trans_\mcC$ is a translation scheme for $\mcR\cup \mcV \cup \mcN_\mcC$ over $\mcV \cup \mcN_\mcC$,
$\redu_\mcC$ is a translation scheme for $\mcR \cup \mcN_\mcC$ over $\mcR\cup \mcV \cup \mcN_\mcC$,
and 
$\lab_\mcC$ is a translation scheme for $\mcC$ over $\mcR \cup \mcN_\mcC$. 

Before we define $\tr_\mcC$, we need to introduce the notion of composition of translation schemes. 
Let $\mcC_0$, $\mcC_1$ and $\mcC_2$ be vocabularies. 
Let $\t_i = \left\langle\phi^i, \psi_C^i: C\in \mcC_i \right\rangle$ be a translation scheme for $\mcC_i$ over $\mcC_{i-1}$.
The {\em\bf composition} of $\t_1$ and $\t_2$, denoted $\t_1\circ\t_2$,
is the translation scheme given by $\t = \left\langle\phi, \psi_C: C\in \mcC_i \right\rangle$
such that 
\begin{enumerate}
 \item $\phi(x) = {(\t^1)}^{\sharp}({(\t^2)}^{\sharp}(x\approx x))$
 \item For every $C\in \unary(\mcC)$, $\psi_C(x) = {(\t^1)}^{\sharp}({(\t^2)}^{\sharp}(C(x)))$
 \item For every $C\in \binary(\mcC)$, $\psi_C = {(\t^1)}^{\sharp}({(\t^2)}^{\sharp}(C(x,y)))$
\end{enumerate}

The translation scheme $\t_\mcC$ is the composition of $\trans_\mcC$, $\redu_\mcC$ and $\lab_\mcC$,
i.e.\ $\t_\mcC = \trans_\mcC\circ(\redu_\mcC\circ \lab_\mcC)$. $\t_\mcC$ is a translation scheme for $\mcC$ over $\mcV\cup \mcN_\mcC$.
As a consequence of Lemma~\ref{lem:translations}, we have:
\begin{lem}\label{ap:lem:final-trans}
\begin{description}
There is a vocabulary $\mcN_\mcC$ consisting of unary relation symbols only, such that for
every $\mcC$-structure $\M$, $\M$ has tree-width at most $k$ iff 
there is an $\mcV\cup \mcN_\mcC$-structure  $\mfA$ 
such that $\t_\mcC^{\star}(\mfA)=\M$ and $\mfA$ is an aligned $k$-tree encoding. 
\end{description}
\end{lem}

\subsection*{Proof of Lemma~\ref{lem:translations-main}}\label{ap:lem:translations-main:finally}
Let $\mcC_\bounded$ and $\mcC_\unbounded$ be vocabularies
such that $\mcC_\bounded\cap \mcC_\unbounded$ only contains unary relation symbols. 
Let $\mcD_\bounded = \mcV\cup \mcN_{\mcC_\bounded}$. 
We define $\tr$ which extends $\t_{\mcC_\bounded}$ as follows: For every $C\in \mcC_\unbounded\backslash \mcC_\bounded$, let $\psi_C(x)=C(x)$ or $\psi_C(x,y)=C(x,y)$. 
I.e.\ $\tr$ copies the relations of $\mcC_\unbounded\backslash \mcC_\bounded$ from the source structure to the target structure without change (except
taking the projection to the universe of the target structure imposed by the formula $\phi(x)$). Let $\mathit{dom} = \enc$ from Lemma~\ref{lem:encoding-axiom}. 
\begin{enumerate}[(a)]
 \item $\phi$ is quantifier-free since $\t_\mcC$ is obtained from the composition
 of translation schemes in which $\phi$ is quantifier-free.
 \item $\psi_C$ is quantifier-free for every $C\in \mcC_\unbounded\backslash \mcC_\bounded$ by the definition in this subsection. 
 $\psi_C$ is quantifier-free for every $C\in \unary(\mcC_\bounded)$ since we have in fact $\psi_C(x)=C(x)$; this is true 
 because it is true for each of $\trans_\mcC$, $\redu_\mcC$ and $\lab_\mcC$. 
 \item For every $C\in \mcC_\bounded$, $\psi_C$ is defined in the translation scheme $\tr_{\mcC_\bounded}$. Since $\tr_{\mcC_\bounded}$
 is a translation scheme for $\mcC_\bounded$ over $\mcV\cup \mcN_{\mcC_\bounded} = \mcD_\bounded$, $\psi_C \in \MS(\mcD_\bounded)$. 
 \item This item follows from Lemma~\ref{ap:lem:final-trans} and using that $\mathit{dom}=\enc$ defines the class of aligned $k$-tree encodings
 on structures $\mfA$ in which $s^\mfA$ is a binary tree. 
\end{enumerate}

\subsection{Proof of  Lemma~\ref{lem:HintikkaTree}}\label{ap:lem:HT}
For a leaf $b$, $\hin_q(\mfT_b)$
depends only on the unary relations which $b$ satisfies. By Theorem~\ref{th:smoothness}, for a vertex $b$
with one child $b_0$ (two children $b_0$, $b_1$), $\hin_q(\mfT_b)$ depends only on the unary relations which $b$ satisfies
and on $\hin_q(\mfT_{b_0})$ (on $\hin_q(\mfT_{b_0})$ and $\hin_q(\mfT_{b_1})$).

Let
\[
\psihin_{\mcC,q}^{\hin}=\part\land\allleaves\land\allinternal_{1}\land\allinternal_{2}
\]
where $\part$ says that $\{C_{\epsilon}:\,\epsilon\in\H_{\mcC,q}\}$
partition the universe, and $\allleaves$, $\allinternal_{1}$, and $\allinternal_{2}$
define the $C_{\epsilon}$ for the leaves respectively the internal
vertices of $\mfT_{0}$ with one or two children.

We give $\part$, $\allleaves$, $\allinternal_1$ and $\allinternal_2$ below.
There are $C^2$ formulas $\leaf(x)$, $\internal_1(x)$, and $\internal_2(x)$, which express that $x$ is a leaf, has one child,
or has two children, respectively.
Let
\[
\part = \forall x \left(\left(\bigvee_{\epsilon\in\H_{\mcC,q}} C_\epsilon(x)\right) \land \left(\bigwedge_{\epsilon_1\not=\epsilon_2\in\H_{\mcC,q}} (\neg C_{\epsilon_1}(x) \lor
\neg C_{\epsilon_2}(x))\right)\right) \,.
\]

For $\mcU\subseteq \unary(\mcC)$,
 let $\mfO_\mcU$ be a $\mcC$-structure with universe $O$ of size $1$ satisfying that
$U^{\mfO_\mcU}=\emptyset$ iff $U\in\mcU$, and $\rt^{\mfO_\mcU} = O$.
Let
\begin{eqnarray*}
\allinternal_1 &=& \forall x\left(\internal_1(x)\to
\bigwedge_{\mcU,\epsilon_1,\epsilon_2}\left(\left(
\thisver_\mcU(x)
\land \child_{\epsilon_2}(x)
\right)\to C_{\epsilon_{1}}(x)\right)\right)\\
 \thisver_\mcU(x)&=&\bigwedge_{U\in\unary(\mcC)} \neg^{(U\in\mcU)}U(x)\\
\child_{\epsilon_2}(y)&=&\exists y\, s(x,y) \land C_{\epsilon_{2}}(y)
\end{eqnarray*}
  $\bigwedge_{\mcU,\epsilon_1,\epsilon_2}$ ranges over $\epsilon_1,\epsilon_2\in\H_{\mcC,q}$ and $\mcU\subseteq \unary(\mcC)$
such that for every structure $\mfA\in\K_\rt$ with
$\hin(\mfA)=\epsilon_2$,
$\hin(\mfO_\mcU \Op \mfA)=\epsilon_1$,
The notation $\neg^{X} Y$ stands for $Y$ if $X$ is false, and for
$\neg Y$ otherwise.

We define
\[\allleaves = \forall x\left(\leaf(x)\to
\bigwedge_{\mcU,\epsilon_1}
\thisver_\mcU(x)\to C_{\epsilon_{1}}(x)\right)
\]
 $\bigwedge_{\mcU,\epsilon_1}$ ranges over $\epsilon_1 \in \H_{\mcC,q}$ and $\mcU\subseteq \unary(\mcC)$, where
$\hin(\mfO_\mcU)=\epsilon_1$.

Finally, $\allinternal_2$ is defined as follows:
\begin{eqnarray*}
\allinternal_2 &=& \forall x (\internal_2(x)\to (\intdist(x) \land \intsame))\\
\intdist(x) &=&
\bigwedge_{\mcU,\epsilon_1,\epsilon_2,\epsilon_3}\left(\left(
\thisver_\mcU(x)
\land \child_{\epsilon_2}(x)\land \child_{\epsilon_3}(x)
\right)\to C_{\epsilon_{1}}(x)\right)\\
\intsame(x) &=&
\bigwedge_{\mcU,\epsilon_1,\epsilon_2}\left(\left(
\thisver_\mcU(x)
\land \twochildren_{\epsilon_2}(x)
\right)\to C_{\epsilon_{1}}(x)\right)\\
\twochildren_{\epsilon_2}(x) &=& \exists^{=2} y\,(s(x,y)\land C_{\epsilon_2}(y))
\end{eqnarray*}
$\bigwedge_{\mcU,\epsilon_1,\epsilon_2,\epsilon_3}$ in $\intdist(x)$ ranges over $\epsilon_1,\epsilon_2,\epsilon_3\in\H_{\mcC,q}$ and $\mcU\subseteq \unary(\mcC)$
such that $\epsilon_2\not=\epsilon_3$, and for every two structure $\mfA,\mfB\in\K_\rt$ with
$\hin(\mfA)=\epsilon_2$ and $\hin(\mfB)=\epsilon_3$,
$\hin((\mfO_\mcU \Op \mfA)\Op \mfB)=\epsilon_1$.
$\bigwedge_{\mcU,\epsilon_1,\epsilon_2}$ in $\intsame(x)$  ranges over $\epsilon_1,\epsilon_2\in\H_{\mcC,q}$ and $\mcU\subseteq \unary(\mcC)$
such that for every structure $\mfA\in\K_\rt$ with
$\hin(\mfA)=\epsilon_2$,
$\hin(\mfO_\mcU \Op \mfA)=\epsilon_1$.

By definition,
$\psihin_{\mcC,q}^{\hin}\in C^2$.
$\mfT_{0}$ is a $\mcC$-structure
in which $s$ is interpreted as a binary tree.
Let $\mfT_1$ be the expansion of $\mfT_0$ such that, for every element $b$ of the universe $T_1=T_0$ of $\mfT_1$,
$b\in C_{\hin(\mfT_b)}^{T_1}$, where $\mfT_b$ is the subtree of $\mfT_0$ rooted at $b$.
In particular, for the root $r$ of $\mfT_0$,
$r\in C_{\hin(\mfT_r)}^{T_1}=C_{\hin(\mfT_0)}^{T_1}$.
We have $\mfT_1 \models \psihin_{\mcC,q}^{\hin}$, and hence (i) holds.
Since $\mfT_1$ is the only expansion of $\mfT_0$ which satisfies  $\psihin_{\mcC,q}^{\hin}$, (ii) holds.

\subsection{Proof of Lemma~\ref{lem:star-type-equivalence}}\label{ap:lem:star-type-equivalence}

We assume $\twoTypeOf(\model_1) = \twoTypeOf(\model_2)$.
Let $u_1,v_1 \in \modelUniverse$ be two elements and let $\twoType = \twoTypeOf^{\model_1}(u_1,v_1)$.
Then there are two elements $u_2,v_2 \in \modelUniverse$ such that $\twoType = \twoTypeOf^{\model_2}(u_2,v_2)$ because of $\twoTypeOf(\model_1) = \twoTypeOf(\model_2)$.
Thus, $\model_1 \models \chi(u_1,v_1)$ implies $\model_2 \models \chi (u_2,v_2)$.
Because $u_1,v_1$ have been chosen arbitrarily, we get that $\model_1 \models \phi$ implies $\model_2 \models \phi$.
The other direction is symmetric.

%
%

\subsection{Proof of Lemma~\ref{lem:chromatic-models}}\label{ap:lem:chromatic-models}

\begin{proof}
Let $z = 2|\messageAlphabet|^2+1$.
We set $\colors(\messageAlphabet) = \{A_i \mid 1 \le i \le z\}$, where the $A_i$ are fresh unary relation symbols.

Let $\model$ be a $\signatureCTwo$-structure with universe $\modelUniverse$.
By Lemma~\ref{lem:coloring}, there is a proper $z$-coloring of $\messageGraph^\model_\messageAlphabet$.
We expand $\model$ to a $(\signatureCTwo \cup \colors(\messageAlphabet))$-structure $\modelAlt$ as follows:
for all elements $u \in \modelUniverse$ we have $A_i^\modelAlt(u)$ iff $u$ has color $i$ in the proper $z$-coloring of $\messageGraph^\model_\messageAlphabet$.
Clearly, $\modelAlt$ is a chromatic $(\signatureCTwo \cup \colors(\messageAlphabet))$-structure.
\end{proof}

\subsection{Proof of Lemma~\ref{lem:coloring}}\label{ap:lem:coloring}
We proceed by induction on the number of vertices $n = |V|$.
The case $n=1$ is trivial.
Assume that $n > 1$ and that the lemma holds for $n-1$.
We have $\sum_{v \in V} \degree(v) = 2 \cdot \sum_{v \in V} \outDegree(v) \le 2nk$.
Thus there is a vertex $v \in V$ with $\degree(v) \le 2k$.
By the induction hypothesis, the graph $G{-}v$ has a proper $2k+1$-coloring.
Because of $\degree(v) \le 2k$ the proper $2k+1$-coloring of $G{-}v$ can be extended to a proper $2k+1$-coloring of $G$.

\subsection{Proof of Lemma~\ref{lem:type-agreement}}\label{ap:lem:type-agreement}

Let $v \in \modelUniverse$ be an element with $\twoTypeOf^{\modelAlt|_\signatureMain}(u,v) = \twoType$.
We get $\modelAlt \models P_{\twoType}(u)$ because of $\modelAlt \models \definitionFormula_\unbounded$.
Because of $\oneTypeOf^{\modelAlt}(u) = \oneTypeOf^{\modelAlt'}(u)$ we have $\modelAlt' \models P_{\twoType}(u)$.
Because of $\modelAlt' \models \definitionFormula_\bounded$ we have $\twoType = \twoTypeOf^{\modelAlt'|_\signatureMain}(u,w)$ for some $w \in \modelUniverse$.
$\rank^i_u(\modelAlt,\modelAlt') = 0$ implies that $\modelAlt' \models \gaifmannRelation_i(u,v)$.
Because of $\modelAlt' \models \functionFormula$ we get $v = w$.
In the same way one can show that $\twoTypeOf^{\modelAlt'|_\signatureMain}(u,v) = \twoType$ implies $\twoTypeOf^{\modelAlt|_\signatureMain}(u,v) = \twoType$.

\subsection{Proof of Lemma~\ref{lem:no-deviation}}\label{ap:lem:no-deviation}

We show $\modelAlt|_\signatureMSO = \modelAlt'|_\signatureMSO$.
Let $u,v \in \modelUniverse$ with $\modelAlt \models \intersectionRelation(u,v)$ for some $\intersectionRelation \in \binary(\signatureMSO)$.
Because of $\modelAlt \models \gaifmannFormula$ we have $\modelAlt \models \gaifmannRelation_i(u,v)$ or $\modelAlt \models \gaifmannRelation_i(v,u)$ for some $i \in [k]$.
Because of $\rank(\modelAlt,\modelAlt') = 0$ we have $\rank^i_u(\modelAlt,\modelAlt') = 0$.
By Lemma~\ref{lem:type-agreement} we have
$\twoTypeOf^{\modelAlt|_\signatureMain}(u,v) = \twoTypeOf^{\modelAlt'|_\signatureMain}(u,v)$.
Thus, $\modelAlt' \models \intersectionRelation(u,v)$.
In the same way one can show that $\modelAlt' \models \intersectionRelation(u,v)$ implies $\modelAlt \models \intersectionRelation(u,v)$.
Thus, $\modelAlt|_\signatureMSO = \modelAlt'|_\signatureMSO$.

We show $\gaif(\modelAlt|_\signatureMSO) = \gaif(\modelAlt')$:
We have that $\gaif(\modelAlt') = \gaif(\modelAlt'|_\signatureOrientedKTree)$ because of $\modelAlt' \models \gaifmannFormula$ and $\modelAlt' \models \containmentFormula$.
We have $\gaif(\modelAlt|_\signatureMSO) = \gaif(\modelAlt|_\signatureOrientedKTree)$ because of $\modelAlt \models \gaifmannFormula$.
We have $\gaif(\modelAlt'|_\signatureOrientedKTree)$ = $\gaif(\modelAlt|_\signatureOrientedKTree)$ because of $\rank(\modelAlt,\modelAlt') = 0$.
Thus, the claim follows.

\subsection{The induced translation \texorpdfstring{$\t^\sharp$}{ }}\label{ap:induced-trans}
We spell out the inductive definition of the induced translation.

Let $\mcC_0$ and $\mcC_1$ be vocabularies.
Given a translation scheme $\t = <\phi, \psi_C: C\in \mcC_0>$ for $\mcC_0$ over $\mcC_1$
we define the induced translation $\t^\sharp$
to be a function from $\MS(\mcC_0)$-formulas to $\mcC_1$-formulas inductively as follows:
\begin{enumerate}
\item
For $C \in \unary(\mcC_0)$ or for monadic second order variables $C$, and for $\theta=C(x)$, we put
\[
\t^\sharp(\theta) = \psi_C(x)\land\phi(x)
\]
\item
For $C \in \binary(\mcC_0)$ and $\theta=C(x,y)$, we put
\[
\t^\sharp(\theta) = \psi_C(x,y)\land\phi(x)\land\phi(y)
\]
\item
For $x \approx y$, we put
\[
\t^\sharp(\theta) = x \approx \land\phi(x)\land\phi(y)
\]
\item For the Boolean connectives the translation distributes, i.e.\
\begin{itemize}
\item if $\theta=\theta_1\lor \theta_2$ then
\[
 t^\sharp(\theta)= (\t^\sharp(\theta_1)\lor \t^\sharp(\theta_2))
\]
\item if $\theta=\neg \theta_1$ then
\[
 t^\sharp(\theta)= \neg \t^\sharp(\theta_1)
\]
\end{itemize}
\item For the existential quantifiers, we relativize to $\phi$:\\
If $\theta = \exists y \, \theta_1$, we put
\[
\t^\sharp(\theta) = \exists y\,( \phi (y) \land \t^\sharp(\theta_1))
\]
If $\theta = \exists  U\, \theta_1$, we put
\[
\t^\sharp(\theta) = \exists U\, (\t^\sharp(\theta_1)\land \forall y\, U(y)\to \phi(y))
\]
\end{enumerate}
 We have somewhat simplified the presentation in~\cite[Definition 2.3]{ar:MakowskyTARSKI} to fit our setting. 

\end{document}